\documentclass[letterpaper]{article} 
\usepackage{aaai23}  
\usepackage{times}  
\usepackage{helvet}  
\usepackage{courier}  
\usepackage[hyphens]{url}  
\usepackage{graphicx} 
\usepackage{natbib}  
\usepackage{caption} 
\usepackage{algorithm}
\usepackage{algorithmic}
\usepackage{mathtools}
\usepackage{xcolor}
\usepackage{threeparttable}

\usepackage{newfloat}
\usepackage{listings}
\DeclareCaptionStyle{ruled}{labelfont=normalfont,labelsep=colon,strut=off} 
\floatstyle{ruled}
\newfloat{listing}{tb}{lst}{}
\floatname{listing}{Listing}
\usepackage{booktabs}       
\usepackage{multirow}
\usepackage{hyperref}
\usepackage{amsthm}

\usepackage{amsmath,amsfonts,bm}

\def\eqref#1{equation~\ref{#1}}

\def\1{\bm{1}}

\def\rvepsilon{{\mathbf{\epsilon}}}

\def\vf{{\bm{f}}}

\def\vp{{\bm{p}}}

\def\vs{{\bm{s}}}
\def\vt{{\bm{t}}}
\def\vu{{\bm{u}}}

\def\vx{{\bm{x}}}
\def\vy{{\bm{y}}}

\def\mD{{\bm{D}}}

\def\mF{{\bm{F}}}

\def\mH{{\bm{H}}}
\def\mI{{\bm{I}}}

\def\mO{{\bm{O}}}

\def\mR{{\bm{R}}}

\def\mX{{\bm{X}}}
\def\mY{{\bm{Y}}}

\DeclareMathAlphabet{\mathsfit}{\encodingdefault}{\sfdefault}{m}{sl}
\SetMathAlphabet{\mathsfit}{bold}{\encodingdefault}{\sfdefault}{bx}{n}

\def\gE{{\mathcal{E}}}
\def\gF{{\mathcal{F}}}
\def\gG{{\mathcal{G}}}

\def\gL{{\mathcal{L}}}

\def\gN{{\mathcal{N}}}
\def\gO{{\mathcal{O}}}

\def\gV{{\mathcal{V}}}

\def\sG{{\mathbb{G}}}

\def\sI{{\mathbb{I}}}

\def\sR{{\mathbb{R}}}

\newcommand{\E}{\mathbb{E}}

\newcommand{\R}{\mathbb{R}}

\newtheorem{theorem}{Theorem}
\newtheorem{lemma}[theorem]{Lemma}
\newtheorem{definition}{Definition}

\newtheorem{proposition}{Proposition}

\title{Energy-Motivated Equivariant Pretraining for 3D Molecular Graphs}

\author{%
    Rui Jiao\textsuperscript{1,2}, Jiaqi Han\textsuperscript{1,2}, Wenbing Huang\textsuperscript{4,5, \thanks{Wenbing Huang and Yang Liu are the corresponding authors.}}, Yu Rong\textsuperscript{6}, Yang Liu\textsuperscript{1,2,3, \footnotemark[1]} \\
}

\affiliations{
    \textsuperscript{1} Beijing National Research Center for Information Science and Technology (BNRist),\\
~~~Department of Computer Science and Technology, Tsinghua University\\
    \textsuperscript{2} Institute for AI Industry Research (AIR), Tsinghua University \\
    \textsuperscript{3} Beijing Academy of Artificial Intelligence \\
    \textsuperscript{4} Gaoling School of Artificial Intelligence, Renmin University of China\\
    \textsuperscript{5} Beijing Key Laboratory of Big Data Management and Analysis Methods\\
    \textsuperscript{6} Tencent AI Lab \\
    \texttt{\{jiaor21,hanjq21\}@mails.tsinghua.edu.cn, hwenbing@126.com,}\\
    \texttt{yu.rong@hotmail.com, liuyang2011@tsinghua.edu.cn}
    }

\begin{document}

\maketitle

\begin{abstract}

Pretraining molecular representation models without labels is fundamental to various applications. Conventional methods mainly process 2D molecular graphs and focus solely on 2D tasks, making their pretrained models incapable of characterizing 3D geometry and thus defective for downstream 3D tasks. In this work, we tackle 3D molecular pretraining in a complete and novel sense. In particular, we first propose to adopt an equivariant energy-based model as the backbone for pretraining, which enjoys the merits of fulfilling the symmetry of 3D space. Then we develop a node-level pretraining loss for force prediction, where we further exploit the Riemann-Gaussian distribution to ensure the loss to be E(3)-invariant, enabling more robustness. Moreover, a graph-level noise scale prediction task is also leveraged to further promote the eventual performance. We evaluate our model pretrained from a large-scale 3D dataset GEOM-QM9 on two challenging 3D benchmarks: MD17 and QM9. Experimental results demonstrate the efficacy of our method against current state-of-the-art pretraining approaches, and verify the validity of our design for each proposed component. Code is available at \url{https://github.com/jiaor17/3D-EMGP}.
\end{abstract}

\section{Introduction}
\label{sec:intro}
Learning informative molecular representation is a fundamental step for various downstream applications, including molecular property prediction~\cite{gilmer2017neural,kearnes2016molecular}, virtual screening~\cite{wallach2015atomnet,zheng2019onionnet}, and Molecular Dynamics (MD) simulation~\cite{chmiela2017machine}. Recent methods exploit Graph Neural Networks (GNNs)~\cite{gilmer2017neural,xu2018powerful} for their power in capturing the topology of molecules, which yet is label-hungry and thus powerless for real scenarios when molecular annotations are unavailable. Therefore, the research attention has been paid to the \emph{self-supervised} pretraining paradigm, to construct the surrogate task by exploring the intrinsic structure within unlabeled molecules. 
A variety of self-supervised works have been proposed, ranging from generative-based models~\cite{kipf2016variational, DBLP:conf/iclr/HuLGZLPL20, DBLP:conf/kdd/HuDWCS20} to contrastive learning~\cite{DBLP:conf/iclr/SunHV020, DBLP:conf/iclr/VelickovicFHLBH19, DBLP:conf/nips/YouCSCWS20, DBLP:conf/icml/YouCSW21, DBLP:journals/corr/abs-2106-04509}. 

In many applications, using and analyzing 3D geometry is crucial and even indispensable; for instance, we need to process 3D coordinates for energy prediction in MD simulation or protein-ligand docking. 
Owing to the fast development in data acquisition, it is now convenient to access large-scale unlabeled molecules with rich 3D conformations~\cite{axelrod2022geom}. It would be quite exciting if we could develop techniques to obtain pretrained models from these unlabeled molecules for 3D tasks with limited data. Nevertheless, existing self-supervised methods~\cite{DBLP:conf/iclr/HuLGZLPL20, DBLP:conf/nips/YouCSCWS20, DBLP:conf/nips/RongBXX0HH20} are weak in leveraging the 3D geometry information. First, 
from the input side, the backbone models they pretrain can only process the input of 2D molecules without the consideration of 3D coordinates. As demonstrated by~\citet{schutt2017schnet}, certain molecular properties (\emph{e.g.} potential energy) are closely related to the 3D structure with which they can be better predicted. Second, for the output side, their pretraining tasks are not 3D-aware, making the knowledge they discover less generalizable in 3D space. Recently, the study by~\citet{DBLP:journals/corr/abs-2110-07728} proposes to impose the 3D information for pretraining; however, its goal is still limited to enhancing 2D models for 2D tasks. 

In this paper,  we investigate 3D molecular pretraining in a complete and novel sense: using 3D backbone models, designing 3D-aware pretraining tasks, and targeting 3D downstream evaluations. However, this is not trivial by any means. The challenges mainly stem from how to maintain the symmetry of our biological world---rotating/translating the 3D conformation of a molecule does not change the law of its behavior.
Mathematically, we should make the backbone E(3)-equivariant, and the pretraining loss E(3)-invariant, where the group E(3) collects the transformations of rotations, reflections, and translations~\cite{satorras2021n}. Unfortunately, typical GNNs~\cite{gilmer2017neural,xu2018powerful} and 3D losses based on Euclidean distance~\cite{luo2020differentiable} do not satisfy such constraints.

To address the above challenges, this paper makes the following contributions: 
\textbf{1.}
We propose an energy-based representation model that predicts E(3)-equivariant force for each atom in the input 3D molecule, by leveraging recently-proposed equivariant GNNs~\cite{satorras2021n, schutt2017schnet, tholke2021equivariant} as the building block. 
\textbf{2.}
To pretrain this model, we formulate a physics-inspired node-level force prediction task, which is further translated to a 3D position denoising loss in an equivalent way. More importantly, we develop a novel denoising scheme with the aid of the proposed Riemann-Gaussian distribution, to ensure the E(3)-invariance of the pretraining task.
\textbf{3.}
We additionally design a graph-level surrogate task on 3D molecules, in line with the observation from traditional 2D methods~\cite{DBLP:conf/nips/RongBXX0HH20} that performing node-level and graph-level tasks jointly is able to promote the eventual performance. For this purpose, we teach the model to identify the noise scale of the input tuple consisting of a clean sample and a noisy one. 
The above ingredients are unified in a general pretraining framework: energy-motivated 3D Equivariant Molecular Graph Pretraining (3D-EMGP).

We pretrain our model on a large-scale dataset with 3D conformations: GEOM-QM9~\cite{axelrod2022geom}, and then evaluate its performance on the two popular 3D tasks: MD17~\cite{chmiela2017machine} and QM9~\cite{ramakrishnan2014quantum}. Extensive experiments demonstrate that our model outperforms state-of-the-art 2D approaches, even though their inputs are augmented with 3D coordinates for fair comparisons. We also inspect how the performance changes if we replace the components of our architecture with other implementations. The results do support the optimal choice of our design.

\section{Related Works}

\paragraph{Self-supervised molecular pretraining} Self-supervised learning has been well developed in the field of molecular graph representation learning. Many pretraining tasks have been proposed to extract information from large-scale molecular dataset, mainly divided into three categories: \emph{contrastive}, \emph{generative} and \emph{predictive}. Contrastive methods aim to maximize the mutual information between different views of the same graph~\cite{DBLP:conf/iclr/SunHV020, DBLP:conf/iclr/VelickovicFHLBH19, DBLP:conf/nips/YouCSCWS20, DBLP:conf/icml/YouCSW21, DBLP:journals/corr/abs-2106-04509}, while generative methods focus on reconstructing the information from different levels of the 2D topological structure~\cite{kipf2016variational, DBLP:conf/iclr/HuLGZLPL20, DBLP:conf/kdd/HuDWCS20}. As for predictive methods, they learn the molecule representation by predicting the pseudo-labels created from the input graphs. For example, GROVER~\cite{DBLP:conf/nips/RongBXX0HH20} proposes to classify the subgraph structure and predict the existence of specific motifs, which leverages domain knowledge into molecule pretraining. However, all the above methods mainly focus on pretraining 2D GNNs without 3D information. More recently, several methods propose to tackle 3D graphs, including 3D Infomax~\cite{DBLP:journals/corr/abs-2110-04126} that maximizes the mutual information between the representations encoded from a 2D and a 3D model, and GraphMVP~\cite{DBLP:journals/corr/abs-2110-07728} which uses contrastive and generative methods to incorporate 3D information into the 2D model. However, the motivation of these two methods remains to benefit 2D models with 3D information. On the contrary, this paper attempts to pretrain 3D models via 3D objectives with the usage for 3D downstream tasks. Besides, ChemRL-GEM~\cite{DBLP:journals/corr/abs-2106-06130} predicts bond length, bond angle prediction, and atom distance for 3D pretraining. While they only employ pairwise or triplet invariant tasks, we formulate both invariant and equivariant pretraining objectives in accordance with our proposed energy-based molecular representation model. \citet{zhu2022unified} formulates the pretraining task as translation between 2D and 3D views, distinct from our goal of realizing pretraining mainly based on 3D conformations. \citet{zhou2022uni} applies masked position denoising, similar to the implemented baseline PosPred in our experiments.

\paragraph{Equivariant graph neural networks} To better analyze the physical and chemical properties of molecules, many researchers regard the molecules as geometric graphs, which additionally assign 3D coordinates on each atom apart from the 2D topological information. Geometric graphs present rotational, translational, and/or reflectional symmetry, as the properties are invariant and the dynamic processes are equivariant to the E(3) or SE(3) transformations in 3D space. To introduce this inductive bias, geometrically equivariant graph neural networks have been proposed to model the geometric graphs. According to~\citet{han2022geometrically}, current 3D GNNs achieve equivariance mainly in three ways: extracting irreducible representations~\cite{thomas2018tensor, DBLP:conf/nips/FuchsW0W20}, utilizing group regular representations~\cite{DBLP:conf/icml/FinziSIW20, hutchinson2021lietransformer} or transforming the 3D vectors into invariant scalars~\cite{schutt2017schnet, satorras2021n, tholke2021equivariant}. Previous works showcase the superiority of equivariant models on several 3D molecular tasks~\cite{tholke2021equivariant, liu2022spherical}, and our goal is to further improve the performance of 3D models via equivariant pretraining on large-scale 3D datasets. 
 
\section{Method}

\begin{figure*}[t]
    \centering
    \includegraphics[width=0.90\textwidth]{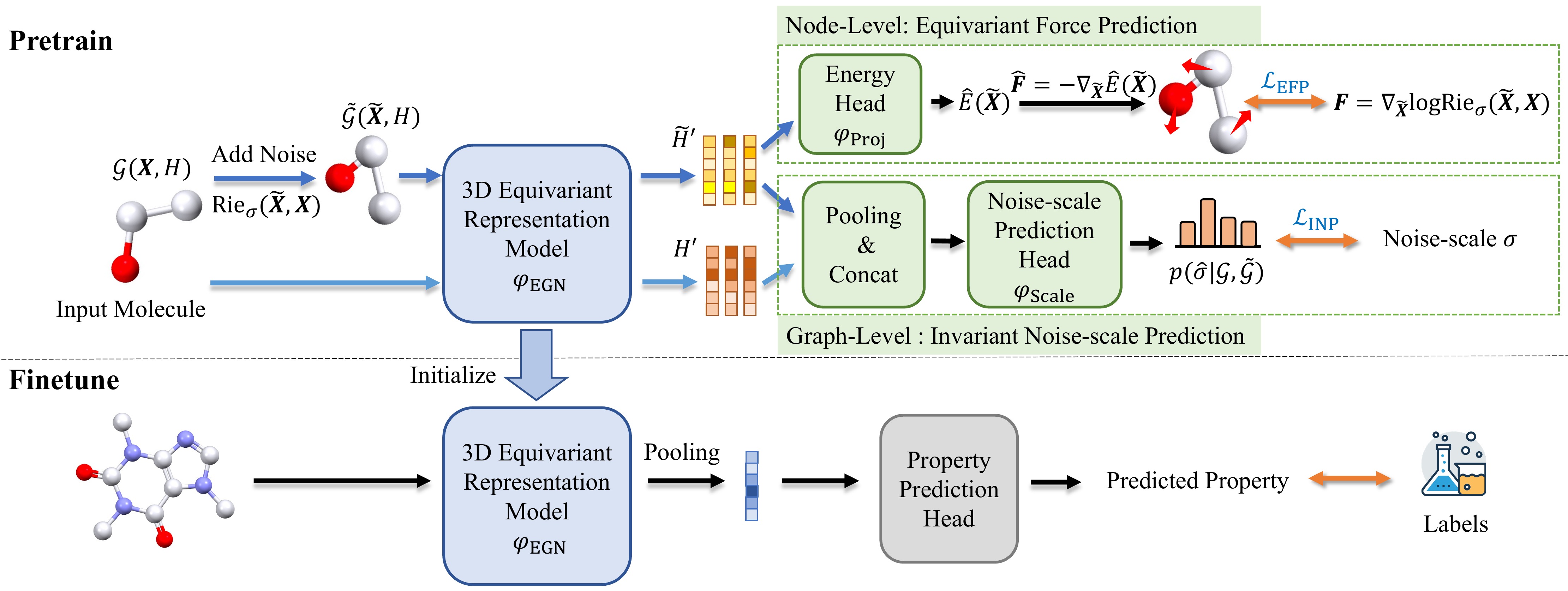}
    \caption{An overview of our 3D-EMGP. It consists of two tasks: node-level equivaraint force prediction and graph-level invariant noise scale prediction. $\mX,\tilde{\mX}$ are the original and perturbed coordinates. $H$ is the input node feature and $H',\tilde{H}'$ are the output features of the original and perturbed graph. $\text{Rie}_{\sigma}(\tilde{\mX}\mid\mX)$ is the proposed Riemann-Gaussian distribution in Eq.~(\ref{eq:rg}).}
    \label{fig:motivation}
    \vskip -0.2in
\end{figure*}


\subsection{Energy-based Molecular Modeling}
\label{sec:ebm}

In general, a molecule consisting of $N$ atoms can be modeled as a \emph{molecular graph} $\gG=(\gV, \gE)$, where $\gV$ is the set of nodes (atoms) and $\gE$ is the set of edges modeling the connectivity, \emph{e.g.}, bond connection or spatial proximity. Each atom is assigned a node feature $h_i,1\leq i\leq N$, representing the properties of the atom, such as atom type. In this work, we elaborate on the 3D information of a molecule, that is, apart from the node attribute $h_i\in\sR^{m}$ for atom $i$, we extra consider the 3D coordinate $\vx_i\in\sR^3$. We denote the configurations of all nodes as $\mX=[\vx_1, \vx_2, \cdots, \vx_N]\in\R^{3\times N}$, and similarly the node features as $H=[h_1, h_2, \cdots, h_N]\in\R^{m\times N}$. Our goal is to pretrain a capable GNN model $\varphi$ that can be generally applied to different downstream 3D tasks, which is depicted as $\varphi(\mX,H,\gE)$.  We hereafter omit the input of $\gE$ for conciseness, unless otherwise specified.

Unlike conventional 2D molecular graph pretraining, we are now provided with vital 3D information, making it possible to leverage the rich geometric context in a unified framework. However, involving 3D conformation is not free of difficulty, and one constraint we should enforce is to make $\varphi$ equivariant for vector outputs, and invariant for scalar outputs. The notion of equivariance/invariance is defined below.
\begin{definition}
\label{De:equ}
A GNN model $\varphi$ is call E(3)-equivariant, if for any transformation $g\in\text{E(3)}$, $\varphi(g\cdot\mX, H)=g\cdot\varphi(\mX,H)$; and it is called E(3)-invariant if $\varphi(g\cdot\mX, H)=\varphi(\mX,H)$. 
\end{definition}

In Definition~\ref{De:equ}, the group action $g\cdot\mX$ is implemented as matrix multiplication $\mO\mX$ for orthogonal transformation $\mO\in\R^{3\times3}$ 
and addition $\mX+\vt$ for translation $\vt\in\R^3$. Basically, for an equivariant function, the output will translate/rotate/reflect in the same way as the input, while for the invariant case, the output stays unchanged whatever group action conducted on the input. Equivariance/invariance is so essential that it characterizes the symmetry of the 3D biochemistry: rotating or translating a molecule will never change its potential energy. There is a variety of E(3)-equivariant GNNs~\cite{han2022geometrically} that can be utilized as our backbone.  By choosing an E(3)-equivariant GNN, $\varphi$ is instantiated as $\varphi_{\text{EGN}}$. Notably, $\varphi_{\text{EGN}}$ is also permutation equivariant regarding the order of the columns in $\mX$ and $H$. 

We now introduce our idea of how to tackle 3D molecular pretraining in a more domain-knowledge-reliable manner. As well studied in biochemistry, the interaction between atoms in 3D space is captured by the forces and potentials, depending on the positions of the atoms, \emph{i.e.}, the molecular conformation. This connection inspires us to incorporate the concept of energy and force into our representation model, making room for designing unsupervised pretraining objectives with 3D geometric information. In light of this, we introduce a representation model that jointly takes into account both energy and force. We denote the energy of a molecule as $E\in\sR$ and the resultant interaction force exerting on atom $i$ as $\vf_i\in\sR^3, 1\leq i\leq N$. The forces over all atoms are collected in the matrix $\mF\in\R^{3\times N}$.
Clearly, $E$ is an invariant graph-level scalar, while $\mF$ consists of equivariant node-level vectors, in terms of the input transformation.

To derive $E$ and $\mF$ by the equivariant model $\varphi_{\text{EGN}}$, we first obtain a node-level 
representation in the latent space:
\begin{align}
\label{eq:H}
    H' = \varphi_{\text{EGN}}\left(\mX, H, \gE\right),
\end{align}
where $H'\in\sR^{k\times N}$ is an invariant representation. Let $\hat{E}, \hat{\mF}$ denote the predicted energy and force. We yield the graph-level energy of the molecule via a graph pooling operation:
\begin{align}
\label{eq:E}
    \hat{E}(\mX) = \varphi_{\text{Proj}}\left(\sum_{i=1}^N h'_i \right),
\end{align}
where $h'_i$ is the $i$-th column of $H'$, $\varphi_{\text{Proj}}: \sR^{k}\mapsto\sR$ is the projection head, realized by a Multi-Layer Perceptron (MLP). Essentially, force corresponds to the direction that causes potential energy to decrease, which implies
\begin{align}
\label{eq:f}
    \hat{\mF}(\mX) &=-\lim_{\Delta \mX\rightarrow 0}\frac{\Delta \hat{E}}{\Delta \mX}= -\nabla_{\mX} \hat{E}(\mX),
\end{align}
where $\nabla_{\mX}$ denotes the gradient \emph{w.r.t.} $\mX$. It is easy to verify that $\hat{E}$ is invariant and $\hat{\mF}$ is equivariant\footnote{More precisely, $\hat{\mF}$ is orthogonality-equivariant but translation-invariant.}.

We attempt to design the first proxy task by leveraging the predicted force $\hat{\mF}$ to fit the force implied in the molecule. However, there is usually no force label provided in the pretraining dataset. Fortunately, we can fulfill this purpose from the lens of node-level denoising---we first add noise to each node's coordinate and then estimate the virtual force that pulls the noisy coordinate back to the original one. We will provide the details in~\textsection~\ref{sec:efp}. Upon the denoising process in the first pretraining task, we further construct a graph-level pretraining objective in an intuitive sense: a desirable model should be able to tell how much noise is added to its input. The detailed strategy is presented in~\textsection~\ref{sec:inp}.

\subsection{Node-Level: Equivariant Force Prediction}
\label{sec:efp}
We start by designing a pretraining objective that well characterizes the 3D geometric information. To fulfill this goal, we resort to the force $\hat{\mF}$ produced by our energy-based molecular model. Yet, it is challenging and non-straightforward to provide a reasonable instantiation, since there is usually no available ground truth force labels in large-scale pretraining datasets. Interestingly, we find a way through by establishing a connection between $\mF$ and the distribution of the conformations $\mX$, and manage to provide the predicted $\hat{\mF}$ with pseudo supervision. The connection is identified by first assuming a Boltzmann energy distribution~\cite{boltzmann1868studien} over the training conformers $\sG$:
\begin{align}
    p\left(\mX\right) &= \frac{1}{Z}\exp\left(-\frac{E(\mX)}{kT} \right),
\end{align}
where $E$ denotes the assumed energy, $k$ is the Boltzmann constant, $T$ is the temperature, and $Z$ is the normalization. By taking the logarithm and computing the gradient over $\mX$, we acquire
\begin{align}
\label{eq:connection}
\nabla_{\mX}  \log p(\mX) \propto -\nabla_{\mX}E(\mX)\coloneqq{\mF}.
\end{align}
It is thus applicable to approach $\hat{\mF}$ by the first term in Eq.~(\ref{eq:connection}), serving as a pseudo force. In light of this, we formulate an equivariant (pseudo) force prediction (EFP) objective over training data $\sG$: 
\begin{align}
\label{eq:efp}
    \gL_{\text{EFP}}=\E_{\gG\sim\sG}\big[  \|\hat{\mF}(\mX)-\nabla_{\mX}\log p(\mX) \|_F^2\big],
\end{align}
where $\hat{\mF}$ is produced by the model $\varphi_{\text{EGN}}$ (Eq.~(\ref{eq:H}-\ref{eq:f})), $\|\cdot\|_F$ computes the Frobenius norm.

Nevertheless, we still have no idea of what the exact form of the data density $p(\mX)$ looks like. Hence it is infeasible to directly apply the loss Eq.~(\ref{eq:efp}). Fortunately, the work by~\citet{vincent2011connection} draws a promising conclusion that Eq.~(\ref{eq:efp}) can be equivalently translated to a denoising problem which is tractable to solve (see Appendix for details). In a nutshell, we instead sample a noisy sample $\tilde{\mX}$ from $\mX$ according to a certain conditional distribution, \emph{i.e.}, $\tilde{\mX}\sim p(\tilde{\mX}\mid\mX)$. Then we substitute the noisy sample into the model $\varphi_{\text{EGN}}$  and perform the replacement of Eq.~(\ref{eq:efp}) by
\begin{align}
\label{eq:new-efp}
    \gL_{\text{EFP-DN}}=\E_{\gG\sim\sG, \tilde{\mX}\sim p(\tilde{\mX}\mid\mX)} \big[  \|\hat{\mF}(\tilde{\mX})-\nabla_{\tilde{\mX}}\log p(\tilde{\mX}\mid\mX) \|_F^2\big].
\end{align}


We now discuss the formulation of the conditional probability $p(\tilde{\mX}\mid\mX)$. 
Different from the traditional denoising process on images or other Euclidean data~\cite{song2020improved, shi2021learning, luo2021predicting}, in our case when considering the 3D geometry, the noise we add should be geometry-aware other than conformation-aware. In other words, $p(\tilde{\mX}\mid\mX)$ should be \emph{doubly E(3)-invariant}, namely,
\begin{align}
    \label{eq:double}
    p(g_1\cdot\tilde{\mX}\mid g_2\cdot\mX)=p(\tilde{\mX}\mid\mX), \forall g_1, g_2\in\text{E}(3),
\end{align}
with the illustrations provided in Fig.~\ref{fig:riemann}.
This is consistent with our understanding: the behavior of molecules with the same geometry should be independent of different conformations. For example, when we rotate the sample $\tilde{\mX}$, the property of $p(\tilde{\mX}\mid\mX)$ by Eq.~(\ref{eq:double}) ensures the loss in Eq.~(\ref{eq:new-efp}) to be unchanged, which is what we desire; similarly, conducting rotation on $\mX$ should also obey the same rule. A conventional choice of $p(\tilde{\mX}\mid\mX)$ is utilizing the standard Gaussian with noise scale $\sigma$: $p(\tilde{\mX}\mid\mX) =  \gN(\mX, \sigma^2\mI)$. This naive form fails to meet the doubly E(3)-invariant property in Eq.~(\ref{eq:double}), which could cause mistaken supervision in Eq.~(\ref{eq:new-efp}). To show this, we derive  $\nabla_{\tilde{\mX}}\log p(\tilde{\mX}\mid\mX)=-\frac{\tilde{\mX}-\mX}{\sigma^2}$ as the force target; if we set $\tilde{\mX}=\mR\mX$ for some rotation matrix $\mR \neq \mI$, then we have $\nabla_{\tilde{\mX}}\log p(\tilde{\mX}\mid\mX) = -\frac{1}{\sigma^2}(\mR - \mI)\mX \neq 0$, which, however, does not align with the true fact that the force between $\tilde{\mX} = \mR\mX$ and $\mX$ should be zero since they share the same geometry. 

\begin{figure}[t!]
    \centering
    \includegraphics[width=0.42\textwidth]{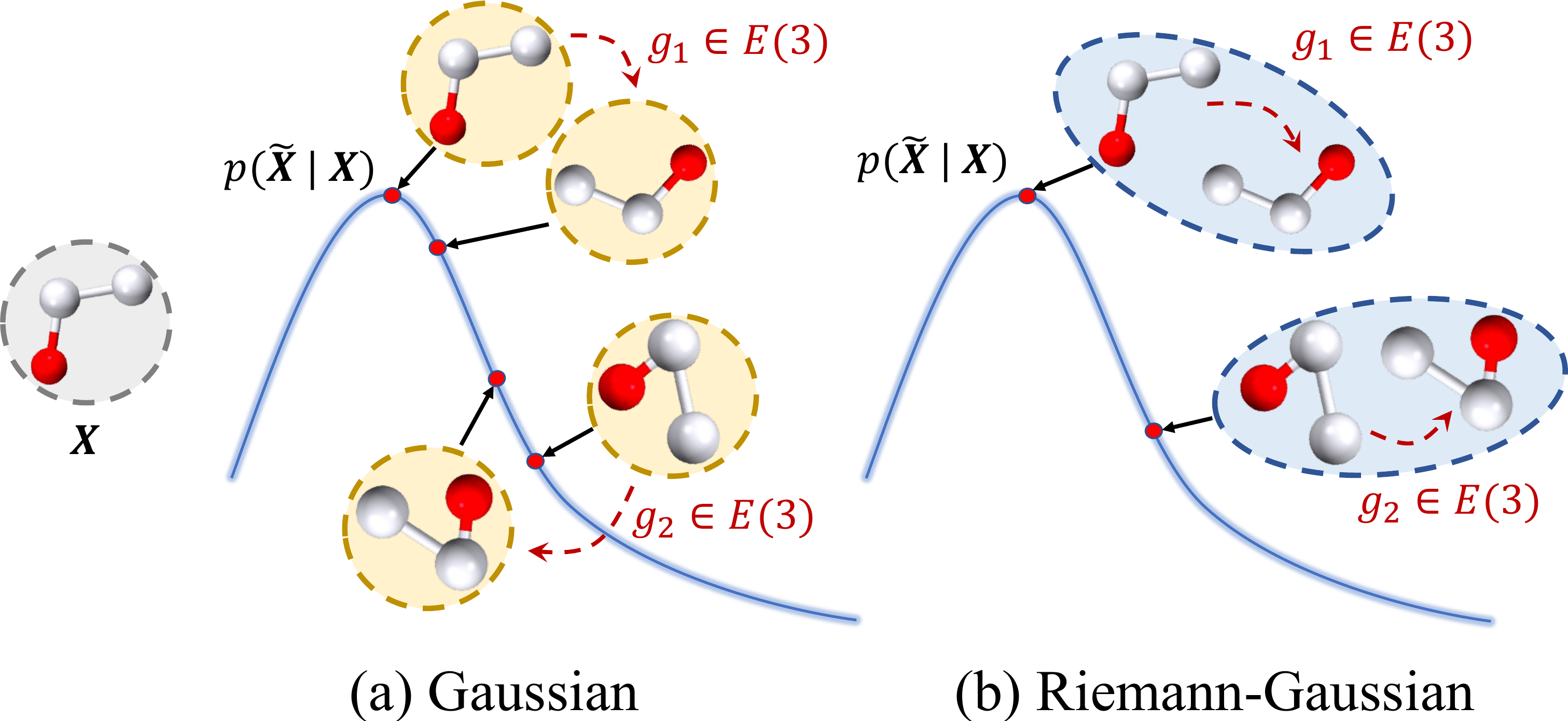}
    \vskip -0.05in
    \caption{Illustration of different distributions. For typical Gaussian, a data point (in dashed circle) is a specific conformation $\mX$, while for Riemann Gaussian, it is a set of conformations with the same geometry $[\mX]\coloneqq\{g\cdot\mX\mid g\in \text{E}(3)\}$.}
    \vskip -0.2in
    \label{fig:riemann}
\end{figure}

To devise the form with the symmetry in Eq.~(\ref{eq:double}), we instead resort to Riemann-Gaussian  distribution~\cite{said2017riemannian} defined as follows: 
\begin{align}
\label{eq:rg}
    p_{\sigma}(\tilde{\mX}\mid\mX)=\text{Rie}_{\sigma}(\tilde{\mX}\mid\mX) \coloneqq \frac{1}{Z(\sigma)}\exp\left(-\frac{d^2(\tilde{\mX},\mX)}{4\sigma^2} \right),
\end{align}
where $Z(\sigma)$ is the normalization term, and $d$ is the metric that calculates the difference between $\tilde{\mX}$ and $\mX$. Riemann-Gaussian is a generalization version of typical Gaussian, by choosing various distances $d$ beyond the Euclidean metric. Here, to pursue the constraint in Eq.~(\ref{eq:double}), we propose to use
\begin{align}
\label{eq:distance}
    d(\mX_1, \mX_2) =\| \mY_1^\top\mY_1 - \mY_2^\top\mY_2\|_F,
\end{align}
where $\mY = \mX - \bm{\mu}(\mX)$ re-positions $\mX$ towards zero mean ($\bm{\mu}(\mX)$ denotes the mean of the columns in $\mX$). One clear benefit is that the distance function $d$ defined in Eq.~(\ref{eq:distance}) satisfies the doubly E(3)-invariance constraint in Eq.~(\ref{eq:double}). Note that $d$ is also permutation invariant with regard to the order of the columns of $\tilde{\mX}$ and $\mX$. We summarize the above discussion as a formal proposition as follows.
\begin{proposition}
For Riemann-Gaussian $\text{Rie}_{\sigma}(\tilde{\mX}\mid\mX)$ defined in Eq.~(\ref{eq:rg}), it is doubly E(3)-invariant as per Eq.~(\ref{eq:double}). 
\end{proposition}

The gradient of Riemann-Gaussian is calculated as follows, with the detailed derivations in Appendix A.2:
\begin{align}
\label{eq:gradient}
     \nabla_{\tilde{\mX}}\log p_{\sigma}(\tilde{\mX}|\mX) = -\frac{1}{\sigma^2}\big[(\tilde{\mY}\tilde{\mY}^\top) \tilde{\mY}-(\tilde{\mY}\mY^\top) \mY \big].
\end{align}
Meanwhile, as proved in Appendix, the calculation in Eq.~(\ref{eq:gradient}) is of the complexity $\gO(N)$, making it computationally efficient for even large-scale molecules.

The last remaining recipe is how to sample $\tilde{\mX}$ from $\mX$ according to the Riemann-Gaussian distribution to provide the input to Eq.~(\ref{eq:new-efp}). It is non-straightforward to accomplish this goal, since the normalization term $Z(\sigma)$ of Riemann-Gaussian is unknown. Here we resort to Langevin dynamics~\cite{schlick2010molecular} which is widely used for approximated sampling when only non-normalized probability density is given. We provide the details in Appendix A.6. Furthermore, to better explore the conformation space, we employ a sampling scheme with multiple levels of noise~\cite{song2020improved}. Particularly, let $\{\sigma_l\}_{l=1}^L$ be a series of noises with different scales. The final EFP loss is given by
\begin{align}
\nonumber
     \gL_{\text{EFP-Final}}&=\E_{\gG\sim\sG, l\sim U(1, L), \tilde{\mX}\sim p_{\sigma_l}(\tilde{\mX}|\mX)}\\ &\big[\sigma_l^2 \|\frac{1}{\sigma_l}\hat{\mF}(\tilde{\mX}) -\frac{1}{\alpha}  \nabla_{\tilde{\mX}}\log p_{\sigma_l}(\tilde{\mX}|\mX) \|_F^2\big],
\end{align}
where $U(1, L)$ is the discrete uniform distribution, and $\nabla_{\tilde{\mX}}\log p_{\sigma_l}(\tilde{\mX}|\mX)$ is provided by Eq.~(\ref{eq:gradient}). We apply a weighting coefficient $\sigma_l^2$ for different noise scales and scale the predicted forces by $1/\sigma_l$ as suggested by~\citet{song2020improved, shi2021learning}. We also add $\alpha$ as a normalization for numerical stability of the inner product; its value is given by $\alpha=(\|\tilde{\mY}\tilde{\mY}^\top\|_F + \|\tilde{\mY}\mY^\top\|_F)/2$ in our experiments. It is proved in Appendix A.5 that the normalization term $\alpha$ also satisfies the doubly E(3)-invariant property.

\begin{table*}[t!]
  \centering
  \resizebox{0.93\linewidth}{!}{
\begin{tabular}{lcccccccc|c}
\toprule
Force   & Aspirin    & Benzene    & Ethanol    & Malon. & Naph. & Salicylic & Toluene    & Uracil     & Average \\
\midrule
Base~\cite{satorras2021n} & 0.3885     & 0.1861     & 0.0599     & 0.1464     & 0.3310     & 0.2683     & 0.1563     & 0.1323     & 0.2086  \\
\midrule
AttrMask~\cite{DBLP:conf/iclr/HuLGZLPL20} & 0.3643     & 0.2277     & 0.0567     & 0.1456     & 0.1773     & 0.3890     & 0.1093     & 0.1560     & 0.2032  \\
EdgePred~\cite{hamilton2017inductive} & 0.4707     & 0.2036     & 0.0743     & 0.1268     & 0.2310     & 0.3400     & 0.1854     & 0.1933     & 0.2281  \\
GPT-GNN~\cite{DBLP:conf/kdd/HuDWCS20} & 0.4278     & 0.2492     & 0.0703     & 0.1484     & 0.2080     & 0.3609     & 0.1541     & 0.2219     & 0.2301  \\
InfoGraph~\cite{DBLP:conf/iclr/SunHV020} & 0.6578     & 0.2743     & 0.1257     & 0.2647     & 0.2860     & 0.5793     & 0.3821     & 0.4238     & 0.3742  \\
GCC~\cite{qiu2020gcc} & 0.3996 &	0.2346 &	0.0662 & 	0.1484 & 	0.2798 & 	0.4263 & 	0.3378 & 	0.2369 & 	0.2662 \\
GraphCL~\cite{DBLP:conf/nips/YouCSCWS20} & \underline{0.2333}     & \underline{0.1845}     & \underline{0.0503}     & 0.0852     & \underline{0.0966}     & \underline{0.1587}     & \underline{0.0725}     & \underline{0.1167}     & \underline{0.1247}  \\
JOAO~\cite{you2021graph} & 0.3646 &	0.2331 & 	0.0642 & 	0.1029 & 	0.2017 & 	0.3020 & 	0.1322 &	0.1683 & 	0.1961 \\
JOAOv2~\cite{you2021graph} & 0.3447 &	0.2198 & 	0.0568 & 	0.0981 & 	0.1889 & 	0.2753 & 	0.1001 &	0.1850 & 	0.1836 \\
GraphMVP~\cite{DBLP:journals/corr/abs-2110-07728} & 0.3198     & 0.2800     & 0.0629     & \underline{0.0788}     & 0.2350     & 0.2641     & 0.0903     & 0.1339     & 0.1831  \\
3D Infomax~\cite{DBLP:journals/corr/abs-2110-04126} & 0.4592 & 0.1914 & 0.0705 & 0.1263 & 0.2642 & 0.3401 & 0.2032 & 0.1836 & 0.2298 \\
GEM~\cite{DBLP:journals/corr/abs-2106-06130} & 0.3994 & 0.2105 & 0.0871 & 0.1161 & 0.1489 & 0.2344 & 0.1193 & 0.1827 & 0.1873 \\
PosPred    & 0.3050     & 0.2023     & 0.0519     & 0.0937     & 0.0971     & 0.2481     & 0.0945     & 0.1270     & 0.1525  \\
\midrule
3D-EMGP & \textbf{0.1560} & \textbf{0.1648} & \textbf{0.0389} & \textbf{0.0737} & \textbf{0.0829} & \textbf{0.1187} & \textbf{0.0619} & \textbf{0.0773} & \textbf{0.0968} \\
\bottomrule
\end{tabular}%
  }
    \vskip -0.05in
    \caption{MAE (lower is better) on MD17 force prediction. All methods share the same backbone as Base.}
  \label{tab:md17}%
      \vskip -0.1in
\end{table*}%

\subsection{Graph-Level: Invariant Noise-scale Prediction}
\label{sec:inp}

In the last subsection, we have constructed a node-level pretraining objective for local force prediction. To further discover global patterns within the input data, this subsection presents how to design a graph-level self-supervised task. Previous studies~\cite{DBLP:conf/iclr/HuLGZLPL20, DBLP:conf/nips/RongBXX0HH20} have revealed for 2D molecules that the node- and graph-level tasks are able to promote each other. Here, we investigate this idea on 3D geometric graphs. 

Recalling that $\tilde{\mX}$ is distributed by $p_{\sigma_l}(\tilde{\mX}|\mX)$, it is expected that a well-behaved model should identify how much the perturbed sample deviates from the original data. Such intuition inspires us to set up a classification problem as noise scale prediction. 
Specifically, our $\varphi_\text{EGN}$ shares the same EGN backbone as in Eq.~(\ref{eq:H}), yielding exactly the same invariant node- and graph-level embeddings. For the input $\mX$ and $\tilde{\mX}$, we first obtain their graph-level embedding $\vu$ and $\tilde{\vu}$ via $\varphi_{\text{EGN}}$, respectively. Instead of using the scalar projection head $\varphi_{\text{Proj}}$ for energy computation, we employ a classification head $\varphi_{\text{Scale}}$ that takes as input a concatenation of the graph-level embeddings of the original conformation $\vu$ and the perturbed conformation $\tilde{\vu}$. The output of $\varphi_{\text{Scale}}$ is the logits $\vp\in\sR^L = \varphi_{\text{Scale}}\left(\vu \| \tilde{\vu} \right)$, where $L$ is the number of noise levels. Finally, a cross-entropy loss is computed between the logits and the label, which is the sampled noise level for the current input. The objective of the invariant noise-scale prediction task is thus given by
\begin{align}
    \gL_{\text{INP}} =\E_{\gG\sim\sG, l\sim U(1, L),\tilde{\mX}\sim  p_{\sigma_l}(\tilde{\mX}|\mX)}\big[ \gL_{\text{CE}}\left(\sI[l], \vp \right)\big],
\end{align}
where $\gL_\text{CE}$ is the cross-entropy loss and $\sI[l]$ is the one-hot encoding of $l$.

Our overall training objective, as illustrated in Fig.~\ref{fig:motivation}, is a combination of both node-level equivariant force prediction loss and graph-level invariant noise scale prediction loss:
\begin{align}
    \gL = \lambda_1\gL_{\text{EFP-Final}} + \lambda_2\gL_{\text{INP}},
\end{align}
where $\lambda_1$ and $\lambda_2$ are the balancing coefficients.

\section{Experiments}

\begin{table*}[t!]
  \centering
  \resizebox{0.93\linewidth}{!}{
\begin{tabular}{lcccccccccccc}
\toprule
   & $\alpha$      & $\Delta_{\epsilon}$        & $\epsilon_{\text{HOMO}}$       &  $\epsilon_{\text{LUMO}}$      & $\mu$         & $C_\nu$         & $G$          & $H$          & $R^2$         & $U$          & $U_0$         & ZPVE \\
\midrule
Base~\cite{satorras2021n} & 0.070       & 49.9       & 28.0         & 24.3       & 0.031      & 0.031      & \underline{10.1}       & 10.9       & \textbf{\underline{0.067}} & 9.7        & \underline{9.3}        & 1.51 \\
\midrule
AttrMask~\cite{DBLP:conf/iclr/HuLGZLPL20} & 0.072      & 50.0         & 31.3       & 37.8       & \underline{0.020}       & 0.062      & 11.2       & 11.4       & 0.423      & 10.8       & 10.7       & 1.90\\
EdgePred~\cite{hamilton2017inductive} & 0.086      & 58.2       & 37.4       & 31.9       & 0.039      & 0.038      & 14.5       & 14.8       & 0.112      & 14.2       & 14.7       & 1.81 \\
GPT-GNN~\cite{DBLP:conf/kdd/HuDWCS20} & 0.103      & 54.1       & 35.7       & 28.8       & 0.039      & 0.032      & 12.2       & 14.8       & 0.158      & 24.8       & 12.0         & 1.75 \\
InfoGraph~\cite{DBLP:conf/iclr/SunHV020} & 0.099      & 72.2       & 48.1       & 38.1       & 0.041      & 0.030       & 16.5       & 14.5       & 0.114      & 14.9       & 16.4       & 1.69 \\
GCC~\cite{qiu2020gcc} & 0.085 &	57.7 &	37.7 &	32.3 &	0.041 &	0.034 &	12.8 &	14.5 &	0.104 &	13.2 &	13.1 &	1.66 \\
GraphCL~\cite{DBLP:conf/nips/YouCSCWS20} & \underline{0.066}      & 45.5       & 26.8       & 22.9       & 0.027      & \underline{0.028}      & 10.2       & 9.6       & 0.095      & \underline{9.7}        & 9.6        & \underline{1.42} \\
JOAO~\cite{you2021graph} & 0.068 &	46.0 & 	28.2 & 	22.8 &	0.028 &	0.030 &	10.5 &	10.0 &	0.076 &	9.9 &	10.1 &	1.48 \\
JOAOv2~\cite{you2021graph} & \underline{0.066} &	45.0 & 	27.8 & 	22.2 &	0.027 &	\underline{0.028} &	9.9 & 	\underline{9.2} &	0.087 &	9.8 &	9.5	& 1.43 \\
GraphMVP~\cite{DBLP:journals/corr/abs-2110-07728} & 0.070       & 46.9       & 28.5       & 26.3       & 0.031      & 0.033      & 11.2       & 10.4       & 0.082      & 10.3       & 10.2       & 1.63 \\
3D Infomax~\cite{DBLP:journals/corr/abs-2110-04126} & 0.075 & 48.8  & 29.8  & 25.7  & 0.034 & 0.033 & 13.0  & 12.4  & 0.122 & 12.5  & 12.7  & 1.67 \\
GEM~\cite{DBLP:journals/corr/abs-2106-06130}   & 0.081 & 52.1  & 33.8  & 27.7  & 0.034 & 0.035 & \multicolumn{1}{r}{13.2} & 13.3  & 0.089 & 12.6  & 13.4  & 1.73 \\
PosPred    & 0.067      & \underline{40.6}       & \underline{25.1}       & \underline{20.9}       & 0.024      & 0.035      & 10.9       & 10.2       & 0.115      & 10.3       & 10.2       & 1.46 \\
\midrule
3D-EMGP & \textbf{0.057} & \textbf{37.1} & \textbf{21.3} & \textbf{18.2} & \textbf{0.020} & \textbf{0.026} & \textbf{9.3} & \textbf{8.7} & 0.092      & \textbf{8.6} & \textbf{8.6} & \textbf{1.38} \\
\bottomrule
\end{tabular}%
}
    \vskip -0.05in
  \caption{MAE (lower is better) on QM9. All methods share the same backbone as Base.}
  \label{tab:qm9main}%
      \vskip -0.25in
\end{table*}%

\subsection{Experimental Setup}
\paragraph{Pretraining dataset} We leverage a large-scale molecular dataset GEOM-QM9~\cite{axelrod2022geom} with corresponding 3D conformations as our pretraining dataset. Specifically, we select the conformations with top-10 Boltzmann weight\footnote{Boltzmann weight is the statistic weight for each conformer determined by its energy.} for each molecule, and filter out the conformations that overlap with the testing molecules in downstream tasks, leading to 100k conformations in total. 
\paragraph{Downstream tasks} To thoroughly evaluate our proposed pretraining framework, we employ the two widely-adopted 3D molecular property prediction datasets: MD17~\cite{chmiela2017machine} and QM9~\cite{ramakrishnan2014quantum}, as the downstream tasks. In detail, MD17 contains the simulated dynamical trajectories of 8 small organic molecules, with the recorded energy and force at each frame. We select 9,500/500 frames as the training/validation set of each molecule. We jointly optimize the energy and force predictions by firstly obtaining the energy and deriving the force by $\mF = -\nabla_{\mX} E $.
QM9 labels 12 chemical properties of small molecules with stable 3D structures. We follow the data split in~\citet{anderson2019cormorant} and~\citet{satorras2021n}, where the sizes of training, validation, and test sets are 100k, 18k, and 13k, respectively.
\paragraph{Baselines}
The baseline without any pretraining is termed as 
\emph{Base}. 
Several widely-used 2D pretraining tasks are evaluated:
\emph{AttrMask}~\cite{DBLP:conf/iclr/HuLGZLPL20} reconstructs the masked atom types; \emph{EdgePred}~\cite{hamilton2017inductive} predicts the existence of chemical bonds; \emph{GPT-GNN}~\cite{DBLP:conf/kdd/HuDWCS20} autoregressively generates the 2D graph in a pre-defined order;
\emph{InfoGraph}~\cite{DBLP:conf/iclr/SunHV020} maximizes the mutual information between node and graph representations;
\emph{GCC}~\cite{qiu2020gcc} employs contrastive learning by the subgraph instance discrimination task;
\emph{GraphCL}~\cite{DBLP:conf/nips/YouCSCWS20} applies contrastive learning on graph representations through several augmentations;
\emph{JOAO} and \emph{JOAOv2}~\cite{you2021graph} further learn to better combine the augmentations. In addition, we also compare with \emph{GraphMVP}~\cite{DBLP:journals/corr/abs-2110-07728} and \emph{3D Infomax}~\cite{DBLP:journals/corr/abs-2110-04126} which simultaneously train 2D- and 3D-GNN models. Notably, different from the original setting in GraphMVP and 3D Infomax, which evaluates the pretrained 2D GNN, we preserve the 3D model for our 3D tasks in the experiments. We further involve \emph{GEM}~\cite{DBLP:journals/corr/abs-2106-06130} which applies bond length prediction, bond angle prediction, and atom distance prediction as 3D pretraining tasks. We also propose \emph{PosPred}, an extension of 2D AttrMask to 3D, as a competitive 3D baseline which masks the positions of a random subset of atoms with the center of each input molecule, and then reconstructs the masked positions.
\emph{For all above model-agnostic methods, we adapt exactly the same 3D backbone as our method, ensuring fairness}. Particularly, we leverage EGNN~\cite{satorras2021n}, a widely adopted equivariant GNN, as our backbone.
Details are deferred to Appendix B.2. 


\subsection{Main Results}
\label{sec:res}
Table~\ref{tab:md17} and~\ref{tab:qm9main} document the results of all pretraining methods on MD17 and QM9, respectively, where the underlined numbers indicate the previous SOTAs on that task, and the numbers in bold are the best results. We interpret the results by answering the questions as follows. 
\begin{itemize}
    \item \emph{How does our 3D-EMGP perform in general?}
It is observed from both Table~\ref{tab:md17} and Table~\ref{tab:qm9main} that 3D-EMGP achieves the best performance in most cases, and its general effectiveness is better justified by checking the average MAE of the last column in Table~\ref{tab:md17}. Particularly for force prediction, the superiority of 3D-EMGP to other methods is more remarkable (3D-EMGP achieves 0.0969, while the second best GraphCL is 0.1247), probably because the design of our node-level force prediction during pretraining is generalizable to the real force distribution after finetuning.
\item \emph{Are the 3D-aware pretraining tasks helpful?}
Compared with Base, 3D-EMGP consistently delivers meaningful improvement on MD17, and gains better performance on QM9 except for the evaluation on $R^2$. We conjecture that the quantity $R^2$ assessing Electronic spatial extent is hardly recovered by the pretraining dataset, hence incurring negative transfer for all pretraining methods. Interestingly, PosPred usually behaves promisingly on MD17, although its 3D prediction objective is simple. 
\item \emph{How do the traditional 2D methods perform on 3D tasks?}
Most 2D methods struggle, especially for force prediction on MD17 and tasks on QM9. Taking MD17 as an example, we observe a serious negative transfer phenomenon on several 2D-pretraining baselines like EdgePred, GPT-GNN, and InfoGraph. This may be due to the dissimilarity between the pretraining 2D task and downstream 3D targets~\cite{Rosenstein05totransfer}. Meanwhile, GraphCL and GraphMVP achieve better results, because GraphMVP has considered the 3D information, while GraphCL employs the contrastive learning technique to capture the invariant information which may exist in both 2D and 3D areas. Our method 3D-EMGP achieves the best performance in most cases, which verifies the effectiveness of our 3D pretraining task. 
\end{itemize}



\subsection{Ablation Studies}
\paragraph{Contribution of each component} We provide extensive ablation results on MD17 to show 
how each component in our model contributes in Table~\ref{tab:abl}. In detail, we study the following aspects. \textbf{1.} We inspect the contributions of the node-level task (\emph{i.e.} EFP) and the graph-level task (\emph{i.e.} INP) by comparing our method with its variants without EFP or INP. It is shown that both EFP and INP improve the performance individually, and their combination leads to more precise predictions. \textbf{2.} To evaluate the importance of the proposed Riemann-Gaussian distribution, we relax the distribution in Eq.~(\ref{eq:rg}) as the Gaussian distribution $p(\tilde{\mX}\mid\mX) =  \gN(\mX, \sigma^2\mI)$, violating the doubly E(3)-invariance in Eq.~(\ref{eq:double}). The results suggest that such relaxation causes certain performance detriment. We also compare with a variant which alternatively applies denoising on the E(3)-invariant distance matrix. This surrogate does not fit in the energy framework, and the performance also drops by a margin. This verifies the empirical significance of leveraging Riemann-Gaussian distribution. \textbf{3.} We analyze the necessity of the energy-based modeling proposed in~\textsection~\ref{sec:ebm}. Instead of deriving the force as the gradient of the energy model, it is also possible to straightly apply the equivariant output from EGNN as the predicted force signal in the EFP loss in Eq.~(\ref{eq:new-efp}). We name this variant as Direct. 
Results in Table~\ref{tab:abl} report that this variant suffers higher MAEs. From an algorithmic point of view, the energy-based strategy is able to better capture the global patterns and therefore lead to preferable performance, by first pooling the embeddings of all nodes as the energy and then computing the gradient of energy as the force.

\begin{table}[t!]
  \centering
     \setlength{\tabcolsep}{4.2pt}
     \small
       \begin{threeparttable}
  \resizebox{0.93\linewidth}{!}{
    \begin{tabular}{l|cccc|cc}
    \toprule
          & \multicolumn{4}{c|}{Proposed Components} & \multicolumn{2}{c}{Average MAE} \\
          & EFP   & INP   & Riemann & Energy & Energy & Force \\
    \midrule
    Base  &       &       &       &       & 0.1191  & 0.2086  \\
    Ours  & \checkmark   & \checkmark   & \checkmark   & \checkmark   & \textbf{0.0876}  & \textbf{0.0968}  \\
    \midrule
    INP only &       & \checkmark   & \checkmark   & \checkmark   & 0.0974  & 0.1350  \\
    EFP only & \checkmark   &       & \checkmark   & \checkmark   & 0.0905  & 0.1193  \\
    Gaussian & \checkmark   & \checkmark   &   &   \checkmark     & 0.0912  & 0.1060  \\
    Distance & \checkmark\tnote{1} & \checkmark &   &   & 0.0931  &  0.1292  \\
        Direct & \checkmark   & \checkmark   &   \checkmark     &   & 0.0914  & 0.1267  \\
    \bottomrule
    \end{tabular}%

    }
        \end{threeparttable}%
        \vskip -0.05in
    \caption{Ablation studies on MD17. $^{1}$Denoising on distance.}
      \label{tab:abl}%
  \vskip -0.25in
\end{table}%

\begin{figure}[htbp]
    \centering
    \vskip -0.1in
    \includegraphics[width=0.40\textwidth]{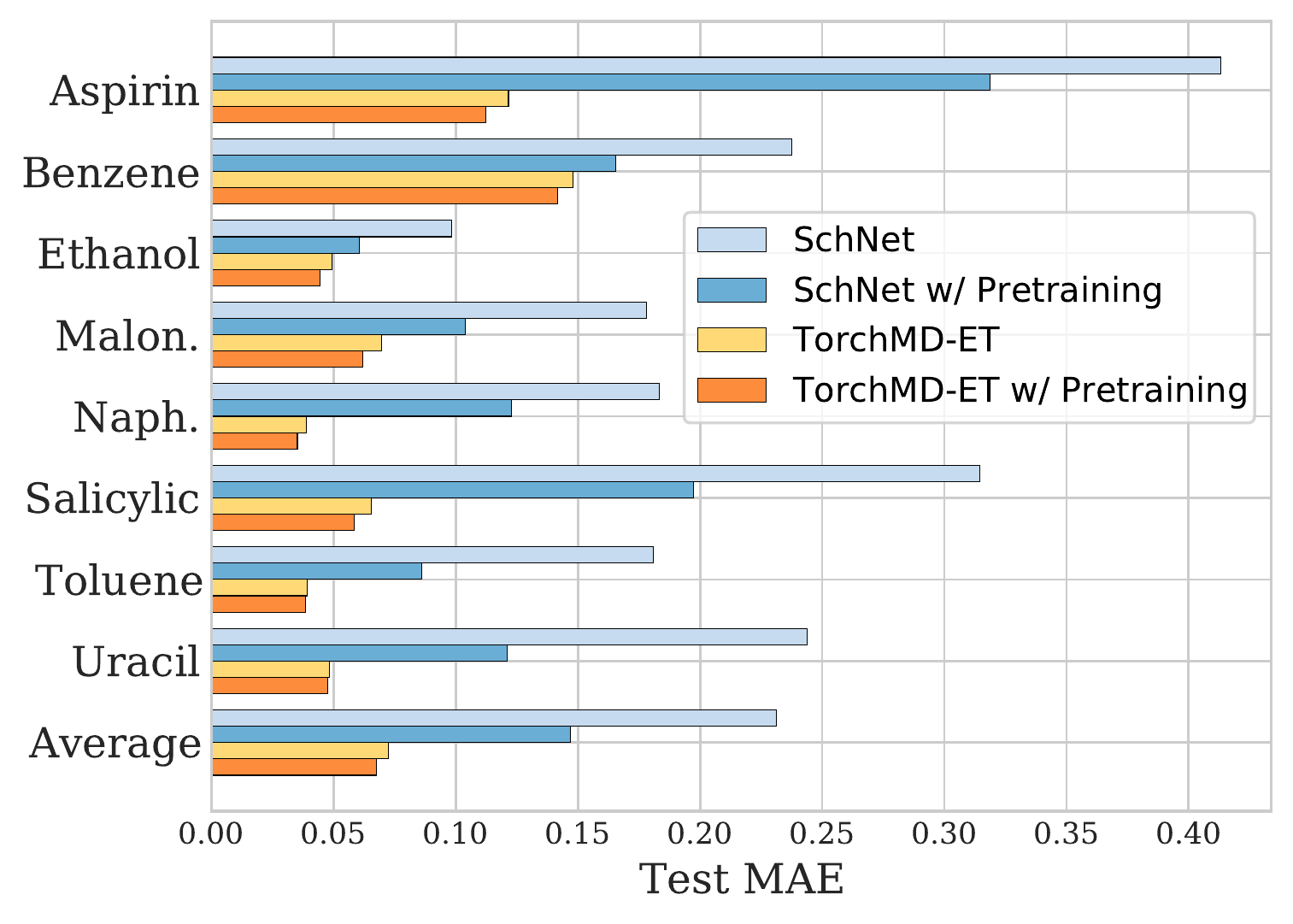}
    \vskip -0.1in
    \caption{MAE on MD17 with different backbones.}
    \vskip -0.2in
    \label{fig:backbone}
\end{figure}

\paragraph{Performance with different backbones} 
We further apply our method to another two 3D backbones, SchNet~\cite{schutt2017schnet} and TorchMD-ET~\cite{tholke2021equivariant} to evaluate the generalization of our proposed self-supervised tasks. Fig.~\ref{fig:backbone} collects the results for force prediction on MD17. The compelling improvement verifies that our pretraining method is widely applicable and generalizes well to a broad family of 3D backbones consistently.

\begin{figure}[htbp]
    \centering
    \vskip -0.1in
    \includegraphics[width=0.40\textwidth]{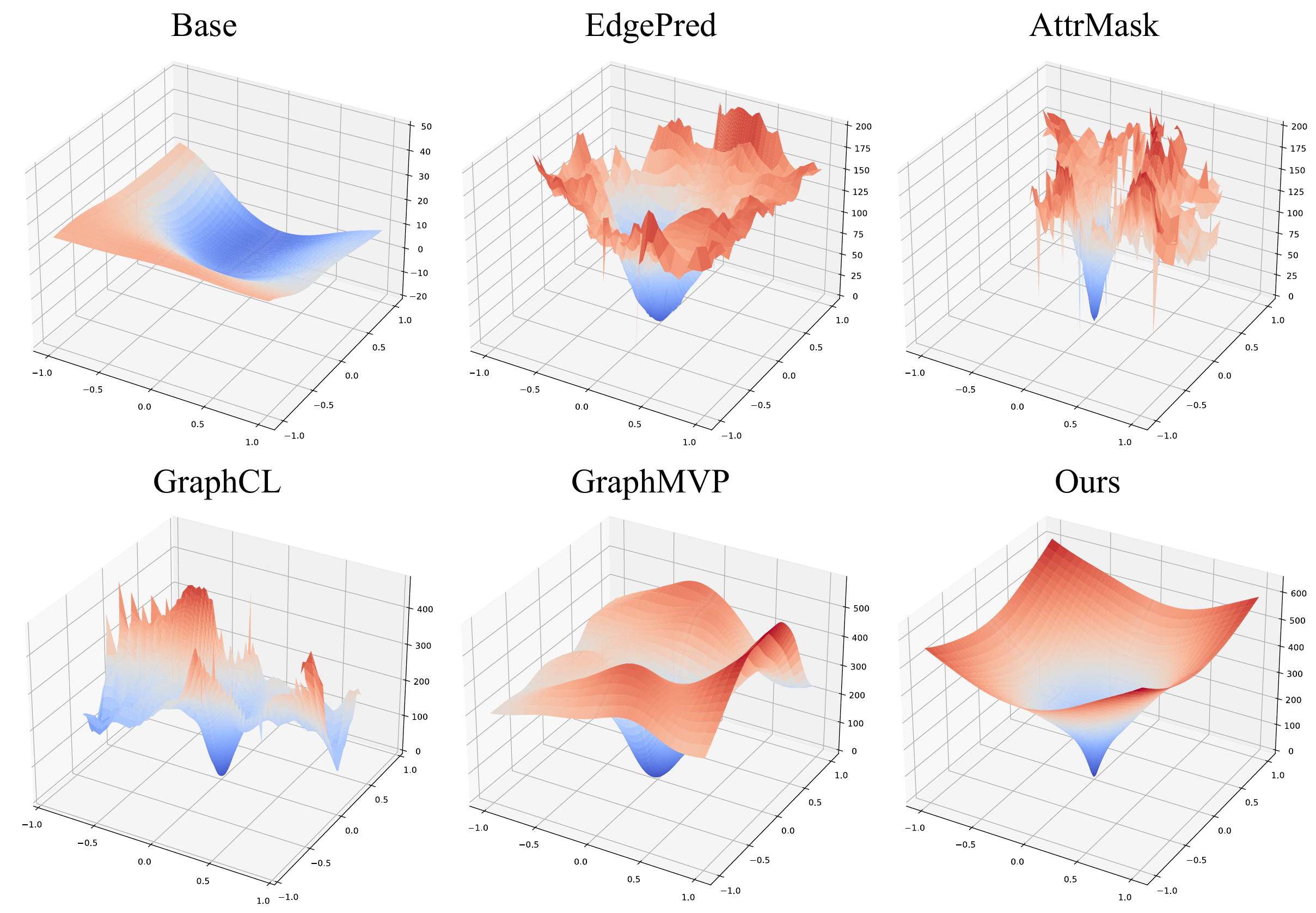}
    \vskip -0.1in
    \caption{Energy landscape of different pretrained models.}
    \vskip -0.2in
    \label{fig:vis}
\end{figure}

\subsection{Visualization}

To probe the representation space of different pretrained models, we visualize the local energy landscape around a given conformation. To do so, we first fix the pretrained representation model and finetune an energy projection head on MD17 to fit ground-truth energy labels, in order to project the pretrained representations onto the energy surface. Note that there is initially no energy projection head for other methods, and we manually add an MLP on top of their backbone models similar to Eq.~(\ref{eq:E}). After training the energy head, we select a random aspirin conformation $\mX$ from MD17 and randomly generate two directions $\mD_1,\mD_2\in\R^{3\times N}$ according to Gaussian distribution. We construct a 2-dimension conformation plane as $\{\tilde{\mX}(i,j)|\tilde{\mX}(i,j)=\mX + i\mD_1 + j\mD_2\}$. For each point by varying the values of $i$ and $j$, we calculate its output energy by $E_{i,j}=E(\varphi_{\text{EGN}}(\tilde{\mX}(i,j)))$,  where $E,\varphi_{\text{EGN}}$ denote the energy projection head and the pretrained model, respectively. Fig.~\ref{fig:vis} plots the energy landscape $(i,j,E_{i,j})$ for several compared approaches and our 3D-EMGP. We interestingly find that the landscape by our method converges towards the original conformation smoothly and decreasingly, which implies the observed conformation corresponds to a metastable state with locally-lowest energy on the projected conformation plane. 
However, the 2D-based pretrained models such as EdgePred, AttrMask, and GraphCL deliver rugged landscapes. We speculate the reason is that their knowledge acquired from the pretraining process does not comply with the underlying energy distribution. The Base method outputs a flat surface, as it is less knowledgeable by solely learning from the small data.

\section{Conclusion}

 In this work, we propose a general self-supervised pretraining framework for 3D tasks on molecules. It consists of a node-level Equivariant Force Prediction (EFP) and a graph-level Invariant Noise-scale Prediction (INP) task to jointly extract the geometric information from a large-scale 3D molecular dataset. Experiments on MD17 and QM9 showcase the superiority of our method to conventional 2D counterparts. Necessary ablations, visualizations, and analyses are also provided to support the validity of our design as well as the generalization of our method on different backbones. 
 
\section*{Acknowledgments}

This work is jointly supported by the National Natural Science Foundation of China
(No.61925601, No.62006137); Guoqiang Research Institute General Project, Tsinghua University (No. 2021GQG1012); Beijing Academy of Artificial Intelligence; Huawei Noah’s Ark Lab;
Beijing Outstanding Young Scientist Program (No. BJJWZYJH012019100020098);Tencent AI Lab. Rhino-Bird Visiting Scholars
Program (VS2022TEG001).

\bibliography{ref}

\appendix
\setcounter{table}{3}
\setcounter{figure}{4}

\newpage

The appendix is organized as follows. We provide a thorough theoretical analysis, including the necessary proofs and derivations in Sec.~\ref{sec:theo}. We describe the implementation details and provide more experiment results and ablations in Sec.~\ref{sec:detail}. Sec.~\ref{sec:vis} displays more visualizations on MD17 and QM9 dataset. 

\section{Theoretical Analysis}
\label{sec:theo}

\subsection{Connection between $\gL_{\text{EFP}}$ and $\gL_{\text{EFP-DN}}$}
\label{sec:a1}
In Sec.~\ref{sec:efp}, we expect to match the predicted forces with the gradient of the log density, which is also called \emph{score}, in Eq.~(\ref{eq:efp1}) as follows.
\setcounter{equation}{5}
\begin{align}
\label{eq:efp1}
    \gL_{\text{EFP}}=\E_{\gG\sim\sG}\big[  \|\hat{\mF}(\mX)-\nabla_{\mX}\log p(\mX) \|_F^2\big].
\end{align}

However, as the close formulation of $p(\mX)$ is unattainable, it is hard to directly train on $\gL_{\text{EFP}}$. Hence, we instead transform the objective into a denoising problem by matching the conditional distribution $p(\tilde{\mX}\mid\mX)$ in Eq.~(\ref{eq:new-efp1}).

\begin{align}
\label{eq:new-efp1}
    \gL_{\text{EFP-DN}}=\E_{\gG\sim\sG, \tilde{\mX}\sim p(\tilde{\mX}\mid\mX)} \big[  \|\hat{\mF}(\tilde{\mX})-\nabla_{\tilde{\mX}}\log p(\tilde{\mX}\mid\mX) \|_F^2\big].
\end{align}

The equivalence between the objectives $\gL_{\text{EFP-DN}}$ and $\gL_{\text{EFP}}$ is ensured by the following proposition.

\setcounter{proposition}{1}
\begin{proposition}
For an arbitrary noise sampling $p(\tilde{\mX}\mid\mX)$,
$\gL_{\text{EFP}}=\gL_{\text{EFP-DN}}+ C$ where $C$ is a constant independent to $\varphi_{\text{EGN}}$, if certain mild conditions hold.
\end{proposition}

Proposition 2 is equivalent to Eq. (11) in~\citet{vincent2011connection}. See Appendix of~\citet{vincent2011connection} for detailed proofs.

\subsection{Calculation of the Gradients of Riemann-Gaussian Distribution}
\label{sec:a2}
According to Eq.~(\ref{eq:rg}) and Eq.~(\ref{eq:distance}), we can derive the formulation of the proposed Riemann-Gaussian distribution as

\begin{align}
    p_{\sigma}(\tilde{\mX}\mid\mX) & = \frac{1}{Z(\sigma)}\exp\left(-\frac{d^2(\tilde{\mX},\mX)}{4\sigma^2} \right) \notag\\
    \nonumber
    &= \frac{1}{Z(\sigma)}\exp\left(-\frac{\| \tilde{\mY}^\top\tilde{\mY} - \mY^\top\mY\|_F^2}{4\sigma^2} \right),
\end{align}

and its partial derivative to $\tilde{\mX}$ is

\begin{align}
\nonumber
   & \nabla_{\tilde{\mX}} \log p_{\sigma}(\tilde{\mX}\mid\mX)\notag \\ =&
     \nabla_{\tilde{\mY}} \log p_{\sigma}(\tilde{\mX}\mid\mX)-\mu(\nabla_{\tilde{\mY}} \log p_{\sigma}(\tilde{\mX}\mid\mX))\notag \\
     =& 
     -\frac{1}{\sigma^2}\big[\tilde{\mY}(\tilde{\mY}^\top \tilde{\mY}-\mY^\top \mY)-\mu(\tilde{\mY}(\tilde{\mY}^\top \tilde{\mY}-\mY^\top \mY))\big]\notag \\
     =& 
     -\frac{1}{\sigma^2}\big[\tilde{\mY}(\tilde{\mY}^\top \tilde{\mY}-\mY^\top \mY)-\mu(\tilde{\mY}\tilde{\mY}^\top \tilde{\mY})-\mu(\tilde{\mY}\mY^\top \mY)\big]\notag \\
     \nonumber
     =&  -\frac{1}{\sigma^2}\left((\tilde{\mY}\tilde{\mY}^\top) \tilde{\mY}-(\tilde{\mY}\mY^\top) \mY \right) \notag,
\end{align}
where, $\mu(\cdot)$ computes the mean of the columns, and $\mu(\tilde{\mY}\tilde{\mY}^\top \tilde{\mY})=\mu(\tilde{\mY}\mY^\top \mY)=0$ since the columns of both $\tilde{\mY}$ and $\mY$ have zero mean. 

\subsection{Complexity Analysis}
\label{sec:a3}
Consider the gradient of the Riemann-Gaussian distribution
\begin{align}
\nonumber
\nabla_{\tilde{\mX}} \log p_{\sigma}(\tilde{\mX}\mid\mX) = -\frac{1}{\sigma^2}\left((\tilde{\mY}\tilde{\mY}^\top) \tilde{\mY}-(\tilde{\mY}\mY^\top) \mY \right) \notag.
\end{align}

We will further prove its $\gO(N)$ complexity, where $N$ is the number of atoms in a molecule.

\textbf{1.} The zero mean process of the coordinates, where $\mY = \mX - \bm{\mu}(\mX), \tilde{\mY} = \tilde{\mX} - \bm{\mu}(\tilde{\mX})$, takes $\gO(N)$ complexity. 

\textbf{2.} Let $\mY = [\vy_1, \vy_2, \cdots, \vy_N]$ and $\tilde{\mY} = [\tilde{\vy}_1, \tilde{\vy}_2, \cdots, \tilde{\vy}_N]$, where $\vy_i, \tilde{\vy}_i \in \sR^3$ denote the original and noised zero mean coordinate of atom $i$. We have $\tilde{\mY}\tilde{\mY}^\top = \sum_{i=1}^N \tilde{\vy}_i \tilde{\vy}_i^\top$ and $\tilde{\mY}\mY^\top = \sum_{i=1}^N \tilde{\vy}_i \vy_i^\top$, where each matrix $\tilde{\vy}_i \tilde{\vy}_i^\top, \tilde{\vy}_i \vy_i^\top \in \sR^{3\times 3}$ can be calculated in constant time. Hence, the calculation of $\tilde{\mY}\tilde{\mY}^\top$ and $\tilde{\mY}\mY^\top$ takes $\gO(N)$ complexity.

\textbf{3.} Consider the vectorized representation of $(\tilde{\mY}\tilde{\mY}^\top) \tilde{\mY}$,
\begin{align}
\nonumber
(\tilde{\mY}\tilde{\mY}^\top) \tilde{\mY} = [(\tilde{\mY}\tilde{\mY}^\top)\tilde{\vy}_1,  (\tilde{\mY}\tilde{\mY}^\top)\tilde{\vy}_2, \cdots, (\tilde{\mY}\tilde{\mY}^\top)\tilde{\vy}_N].
\end{align}
As $\tilde{\mY}\tilde{\mY}^\top \in \sR^{3\times 3}$ and $\tilde{\vy}_i \in \sR^3, \forall 1\leq i\leq N$, the calculation of $(\tilde{\mY}\tilde{\mY}^\top)\tilde{\vy}_i$ costs constant time, and $(\tilde{\mY}\tilde{\mY}^\top) \tilde{\mY}$ can be calculated in $\gO(N)$ complexity. Similarly, the calculation of $(\tilde{\mY}\mY^\top) \mY$ also costs $\gO(N)$ time.

Overall, the calculation of $\nabla_{\tilde{\mX}} \log p_{\sigma}(\tilde{\mX}\mid\mX)$ takes $\gO(N)$ complexity.

\subsection{Proof of Proposition 1}
\label{sec:a4}
\setcounter{proposition}{0}

\begin{proposition}
For Riemann-Gaussian $\text{Rie}_{\sigma}(\tilde{\mX}\mid\mX)$ defined in Eq.~(\ref{eq:distance}), it is doubly E(3)-invariant in accordance to Eq.~(\ref{eq:double}). 
\end{proposition}

To prove Proposition 1, we first introduce the following lemma.


\begin{lemma}
Let $\mY = \mX - \mu(\mX)$. Then $\gF(\mX) =  \mY^\top \mY$ is E(3)-invariant, \emph{i.e.},  $\mY_1^\top \mY_1 = \mY_2^\top \mY_2$ if  $\mX_1 = g\cdot\mX_2, \forall g\in$ E(3).
\end{lemma}

\begin{proof} Consider $\mX_1 = g\cdot \mX_2=\mO\mX_2 + \vt$, where $\mO \in \sR^{3\times3} $ is an orthogonal matrix which represents rotation/reflection and $\vt\in \sR^3$ is the translation vector. For $\mY_1, \mY_2$, we have
\begin{align}
\nonumber
    \mY_1 &= \mX_1 - \mu(\mX_1) = \mO\mX_2 + \vt - \mu(\mO\mX_2 + \vt) \notag \\
    \
    \nonumber
    &= \mO\mX_2 - \mu(\mO\mX_2) = \mO(\mX_2 - \mu(\mX_2)) = \mO\mY_2.
\end{align}

That is, the translation transformation is reduced by the zero mean operation. Moreover, for the inner-product computation, we have
\begin{align}
\nonumber
    \gF(\mX_1) = \mY_1^\top \mY_1 = \mY_2^\top \mO^\top \mO \mY_2 = \mY_2^\top \mY_2 = \gF(\mX_2).
\end{align}
\end{proof} 
Based on the above lemma, we can directly prove Proposition 1 as follows.
\begin{proof} \begin{align}
   & p_{\sigma}(g_1\cdot\tilde{\mX}\mid g_2\cdot\mX)\notag \\ =& \frac{1}{Z(\sigma)}\exp\left(-\frac{d^2(g_1\cdot\tilde{\mX},g_2\cdot\mX)}{4\sigma^2} \right) \notag\\
    =& \frac{1}{Z(\sigma)}\exp\left(-\frac{\|\gF(g_1\cdot\tilde{\mX})-\gF(g_2\cdot\mX))\|_F^2}{4\sigma^2} \right) \notag\\
    =& \frac{1}{Z(\sigma)}\exp\left(-\frac{\|\gF(\tilde{\mX})-\gF(\mX))\|_F^2}{4\sigma^2} \right) \notag\\
    \nonumber
    =& p_{\sigma}(\tilde{\mX}\mid \mX),
\end{align}

which proves the doubly E(3)-invariance of the Riemann-Gaussian distribution.
\end{proof}

\subsection{Discussion of Normalization Term}
\label{sec:a5}
The inner product term $\tilde{\mY}\tilde{\mY}^\top$ and $\tilde{\mY}\mY^\top$ in Eq.~(\ref{eq:gradient}) can be very large and make the pretraining procedure numerically unstable. Therefore, we adopt the normalization term $\alpha = (\|\tilde{\mY}\tilde{\mY}^\top\|_F + \|\tilde{\mY}\mY^\top\|_F)/2$ in the final EFP loss to ensure the the pseudo force achieve the same order of magnitude as the coordinates $\tilde{\mY}, \mY$. To align with our analysis, this normalization term, while improving numerical stability by controlling the magnitude of the inner product, should also satisfy the doubly-E(3) invariance. Formally, we have

\setcounter{proposition}{2}
\begin{proposition}
Let $\tilde{\mY}=\tilde{\mX}-\mu(\tilde{\mX}), \mY=\mX-\mu(\mX)$, and $\alpha(\tilde{\mX},\mX)=(\|\tilde{\mY}\tilde{\mY}^\top\|_F + \|\tilde{\mY}\mY^\top\|_F)/2$, $\alpha$ satisfy the doubly-E(3) invariance, \emph{i.e.}, $\alpha(g_1\cdot \tilde{\mX}, g_2\cdot \mX) = \alpha(\tilde{\mX},\mX), \forall g_1, g_2 \in \text{E(3)}$.
\end{proposition}
\begin{proof}
\begin{align}
\nonumber
& \alpha(g_1\cdot \tilde{\mX}, g_2\cdot \mX) \\
=& \alpha(\mR_1\tilde{\mX} + \vt_1, \mR_2 \mX+\vt_2)\notag \\
=& (\|\mR_1\tilde{\mY}\tilde{\mY}^\top \mR_1^\top\|_F + \|\mR_1\tilde{\mY}\mY^\top \mR_2^\top\|_F)/2 \notag \\
=& (\|\tilde{\mY}\tilde{\mY}^\top\|_F + \|\tilde{\mY}\mY^\top\|_F)/2 \notag\\
=& \alpha(\tilde{\mX},\mX)\notag
\end{align}
\end{proof}

\subsection{Sampling from Riemann-Gaussian Distribution}
\label{sec:a6}
We adopt Langevin Dynamics~\cite{schlick2010molecular} to approximately sample noisy coordinates $\tilde{\mX}$ from $\mX$ according to the proposed Riemann-Gaussian distribution. The detailed algorithm is shown as follows.

For implementation, we find that the model performance is not sensitive to $T$, and choosing $T=1$ is already able to yield favorable results. See Appendix~\ref{sec:abl_apd} for more details.

\begin{algorithm}[htbp]
  \caption{Riemann-Gaussian Sampling via Langevin Dynamics}
  \label{alg:rg_sample}
\begin{algorithmic}
  \STATE {\bfseries Input:} Original coordinates $\mX$, Noise scale $\sigma$,Maximum iteration step $T$.
  \STATE Initialize $\tilde{\mX}_0 \leftarrow \mX$
  \FOR{step $t=1$ to $T$}
  \STATE $\beta_t \leftarrow \frac{\sigma ^ 2}{2 ^ t}$;
  \STATE Derive $\vs_t\leftarrow  \nabla_{\tilde{\mX}_{t-1}}\log \text{Rie}_{\sigma}(\tilde{\mX}_{t-1}|\mX)$ according to Eq.~(11);
  \STATE Sample  $\rvepsilon_t \sim \gN(0,\mI)$;
  \STATE Calculate normalization term $\alpha_t = (\|\tilde{\mY}_{t-1}\tilde{\mY}_{t-1}^\top\|_F + \|\tilde{\mY}_{t-1}\mY^\top\|_F)/2$;
  \STATE $\tilde{\mX}_t = \tilde{\mX}_{t-1} + \beta_t \vs_t / \alpha_t + \sqrt{2\beta_t}\rvepsilon_t$;
  \ENDFOR
  \STATE {\bfseries Output:} Sampled coordinates $\tilde{\mX}_T$;
\end{algorithmic}
\end{algorithm}

\section{Experimental Details}
\label{sec:detail}

\subsection{Detailed Implementations of Baselines}
\label{sec:detail_base}

\paragraph{3D Graph Representation} As mentioned in Sec.~\ref{sec:ebm}, a 3D molecule can be modeled with node features $\mH$, edge connections $\gE$ and coordinates $\mX$. For implementation, we represent a 3D molecule as a fully connected graph, and take atom types as the node features. The edge feature $e_{ij}\in\{\text{1-hop},\text{2-hop},\text{3-hop},\text{others}\}$ is determined by the minimum hops between atom $i,j$ and models the bond, angle, torsion and non-bond interactions in molecules similar to previous works~\cite{klicpera2019directional, luo2021predicting, ying2021transformers}. The 1-hop connections are determined by the covalent bonds.

\paragraph{Implement 2D SSL baselines with 3D graphs} We adopt the above 3D graph representation method to all baseline methods and our proposed 3D-EMGP for a fair comparison. Some of the 2D baselines require the 2D topology. For EdgePred, we predict whether there exists a covalent bond between node pair $i,j$ in a molecule. For GPT-GNN, we determine the node order via BFS on the 2D edges. For GraphCL, we apply AttrMask, EdgePert, Subgraph, and NodeDrop on the 2D graph, and transform the augmented 2D graph into the 3D setting, i.e., fully connect the graph and determine the edge types with $\{\text{1-hop, 2-hop, 3-hop, others}\}$.

\paragraph{Implementation of PosPred} Given the input coordinates of a conformation $\mX$, we first randomly select $K\%$ of the atoms, replacing their coordinate with the center of the molecule $\mu(\mX)$. This gives a perturbed conformation $\mX' = \mX \odot (1-M) + \mu(\mX) \odot M$, where $M=[m_1, m_2, \cdots, m_N]$ denotes the mask vector. $m_i=1$ if atom $i$ is selected otherwise 0. The training objective of PosPred is to predict the original coordinates of the selected atoms, i.e., $\mathcal{L} = \|(\hat{\mX} - \mX)\odot M  \|_F $, where $\hat{\mX} = f_{\text{EGNN}}(\mX', H, \mathcal{G})$ is the model predicted coordinate matrix.

\subsection{Hyper-parameter Settings}
\label{sec:hyper}
We apply the EGNN model with 7 layers, 128 hidden dimensions for all methods during the pretraining and finetuning procedure. The atom coordinates are fixed among layers. For all pretraining methods, we train the model with epoch 300, batch size 64$\times$4 GPUs, Adam optimizer, and cosine decay with an initial learning rate $5\times 10^{-4}$. For our 3D-EMGP, we select $\sigma_1=10, \sigma_L=0.01, L=50$, uniformly decreasing in log scale, similar to previous generation methods~\cite{shi2021learning, luo2021predicting} and $\lambda_1=1.0, \lambda_2=0.2$ to balance the EFP and INP tasks. For AttrMask, PosPred, GraphCL, and GraphMVP, which require mask operations, we set the mask ratio as 0.15. For GraphMVP, we adopt a GIN model with 5 layers, 128 hidden dimensions as the 2D-GNN side. For 3D Infomax, we select a 7-layer PNA model with 128 hidden dimensions as the 2D side. For GEM, we apply the proposed bond length prediction, bond angle prediction, and atom distance prediction as the SSL tasks, and the training objective is the sum of the three tasks. For QM9, we finetune the model on each individual task with epoch 1,000, batch size 96, Adam optimizer with weight decay $1\times10^{-16}$, and cosine decay with initial learning rate $5\times10^{-4}$. For MD17, we jointly learn the energy and forces of each molecule trajectory with epoch 1,500, batch size 100, Adam optimizer with weight decay $1\times10^{-16}$, and ReduceLROnPlateau scheduler with patience 30, decay factor 0.9, and initial learning rate within $\{1\times10^{-3}, 5\times10^{-4}\}$. We set $\lambda_{\text{energy}}=0.2, \lambda_{\text{force}}=0.8$ to balance the loss of energy and forces. All training procedures are conducted on NVIDIA Tesla V100 GPUs.

\subsection{Results for MD17 energy prediction tasks}

\begin{table}[htbp]
  \centering
    \setlength{\tabcolsep}{2pt}
  \caption{MAE (lower is better)  on MD17 for energy prediction tasks. All methods share the same backbone as Base.}
  \resizebox{\linewidth}{!}{
\begin{tabular}{lcccccccc|c}
\toprule
Energy           & Aspirin    & Benzene    & Ethanol    & Malon. & Naph. & Salicylic & Toluene    & Uracil     & Average \\
\midrule
Base & 0.2044     & 0.0755     & 0.0532     & 0.0748     & 0.1961     & 0.1289     & 0.1036     & 0.1165     & 0.1191  \\
\midrule
AttrMask & 0.1951     & 0.0709     & 0.0495     & 0.0796     & 0.1239     & 0.1531     & 0.0924     & 0.1066     & 0.1089  \\
EdgePred & 0.2232     & 0.0717     & 0.0503     & 0.0739     & 0.1428     & 0.1346     & 0.1018     & \underline{\textbf{0.0978}} & 0.1120  \\
GPT-GNN & 0.1656     & 0.0720     & 0.0488     & 0.0740     & 0.1370     & 0.1460     & 0.0929     & 0.1040     & 0.1050  \\
GCC & 0.1897 &	0.0725 &	0.0499 &	0.0762 &	0.1455 &	0.1649 &	0.1457 &	0.1094 & 0.1192 \\
InfoGraph & 0.3320     & 0.0817     & 0.0577     & 0.0859     & 0.1356     & 0.2122     & 0.1261     & 0.1392     & 0.1463  \\
GraphCL & \underline{0.1299}     & 0.0706     & 0.0492     & 0.0722     & 0.1267     & 0.1187     & 0.0901     & 0.1049     & 0.0953  \\
JOAO &	0.1691 &	0.0708 &	0.0482 &	\underline{0.0696} &	0.1397 &	0.1232 &	0.0946 &	0.1115 &	0.1033 \\
JOAOv2 &	0.1770 &	0.0725 &	0.0527 &	0.0698 &	0.1455 &	0.1381 &	0.0932 &	0.1094 &	0.1073 \\
GraphMVP & 0.1575     & 0.0853     & \underline{\textbf{0.0479}} & 0.0726     & 0.2315     & 0.1375     & 0.0964     & 0.1041     & 0.1166  \\
3D Infomax & 0.2628 & 0.0729 & 0.0508 & 0.0757 & 0.1745 & 0.1755 & 0.1228 & 0.1017 & 0.1296 \\
GEM & 0.1657 & 0.0716 & 0.0505 & 0.0735 & 0.1209 & 0.1138 & 0.0951 & 0.1051 & 0.0995 \\
PosPred    & 0.1470     & \underline{0.0677} & 0.0497     & 0.0748     & \underline{\textbf{0.1101}} & \underline{0.1161}     & \underline{0.0902} & 0.1007     & 0.0945  \\
\midrule
3D-EMGP & \textbf{0.1118} & \textbf{0.0671}     & 0.0484     & \textbf{0.0693} & 0.1107     & \textbf{0.1058} & \textbf{0.0874}     & 0.1001    & \textbf{0.0876} \\
\bottomrule
\end{tabular}%

  }
  \label{tab:md17_apd}%
  \vskip -0.1in
\end{table}%

We simultaneously learn the energy and forces of each molecule trajectory on MD17. Table~\ref{tab:md17_apd} shows the corresponding results on energy prediction tasks with Table~\ref{tab:md17} in Sec.~\ref{sec:res}, which consistently indicates the superiority of our proposed 3D-EMGP.

\subsection{Complete Results for Ablation Studies}
\label{sec:abl_apd}

\begin{table}[htbp]
  \centering
  \setlength{\tabcolsep}{2.8pt}
  \caption{Complete results for different variants of 3D-EMGP.}
  \resizebox{\linewidth}{!}{
    \begin{tabular}{lccccccccc}
    \toprule
    Energy & Aspirin & Benzene & Ethanol & Malon. & Naph. & Salicylic & Toluene & Uracil & Average \\
    \midrule
    Base  & 0.2044  & 0.0755  & 0.0532  & 0.0748  & 0.1961  & 0.1289  & 0.1036  & 0.1165  & 0.1191  \\
    \midrule
    Ours  & \textbf{0.1118} & 0.0671  & \textbf{0.0484} & \textbf{0.0693} & \textbf{0.1107} & 0.1058  & \textbf{0.0874} & 0.1001  & \textbf{0.0876} \\
    EFP only & 0.1180  & 0.0695  & 0.0485  & 0.0713  & 0.1156  & 0.1111  & 0.0903  & 0.0999  & 0.0905  \\
    INP only & 0.1593  & 0.0684  & 0.0494  & 0.0750  & 0.1217  & 0.1199  & 0.0892  & \textbf{0.0962} & 0.0974  \\
    \midrule
    Gaussian & 0.1214  & 0.0759  & 0.0499  & 0.0716  & 0.1176  & \textbf{0.1050} & 0.0895  & 0.0984  & 0.0912  \\
    Distance & 0.1217 &	0.0732 & 0.0504 & 0.0710 & 0.1131 & 0.1202 & 0.0936 & 0.1021 & 0.0932 \\
    \midrule
    Direct & 0.1179  & \textbf{0.0661} & 0.0498  & 0.0725  & 0.1161  & 0.1163  & 0.0917  & 0.1006  & 0.0914  \\
    \bottomrule
    \toprule
    Force & Aspirin & Benzene & Ethanol & Malon. & Naph. & Salicylic & Toluene & Uracil & Average \\
    \midrule
    Base  & 0.3885  & 0.1861  & 0.0599  & 0.1464  & 0.3310  & 0.2683  & 0.1563  & 0.1323  & 0.2086  \\
    \midrule
    Ours  & \textbf{0.1560} & \textbf{0.1648} & 0.0389  & 0.0737  & \textbf{0.0829} & 0.1187  & 0.0619  & 0.0773  & \textbf{0.0968} \\
    EFP only & 0.1808  & 0.2211  & 0.0389  & 0.0773  & 0.0929  & 0.1773  & 0.0615  & 0.1043  & 0.1193  \\
    INP only & 0.2709  & 0.1704  & 0.0397  & 0.0756  & 0.1373  & 0.2184  & 0.0599  & 0.1080  & 0.1350  \\
    \midrule
    Gaussian & 0.1693  & 0.2361  & \textbf{0.0366} & \textbf{0.0692} & 0.0917  & \textbf{0.1164} & \textbf{0.0588} & \textbf{0.0697} & 0.1060  \\
    Distance & 0.2408 & 0.1861 & 0.0502 & 0.0854 & 0.1170 & 0.1595 & 0.0801 & 0.1148 & 0.1292 \\
    \midrule
    Direct & 0.2263  & 0.1649  & 0.0422  & 0.0952  & 0.1058  & 0.1854  & 0.0741  & 0.1194  & 0.1267  \\
    \bottomrule
    \end{tabular}%
   }
  \label{tab:abl_all}%
  \vskip -0.1in
\end{table}%

\paragraph{Contribution of each component} We implement 4 variants of the original 3D-EMGP to explore the contribution of each component. \textbf{EFP only} and \textbf{INP only} adjust the weight to $\lambda_1=1, \lambda_2=0$ and $\lambda_1=0, \lambda_2=1$. \textbf{Gaussian} changes the force target in Eq.~(11) into $\nabla_{\tilde{\mX}}\log p(\tilde{\mX}\mid\mX)=-\frac{\tilde{\mX}-\mX}{\sigma^2}$, and \textbf{Direct} updates the coordinates in the last layer of EGNN and predicts the atom forces from the difference of the output coordinates and the input one, \emph{i.e.}, $\mF=\frac{1}{\sigma}(\varphi_{\text{EGN}}(\mX,H)-\mX)$, where $\frac{1}{\sigma}$ is a normalization term to balance different noise scale. \textbf{Distance} adds Gaussian noises on distance matrix and requires the model to predict the noise applied on each pairwise distance. Detailed results for MD17 energy and force prediction tasks are shown in Table~\ref{tab:abl_all}, where the original model achieves better averaged MAEs than the variants. The results indicate the effectiveness of each component in our proposed 3D-EMGP.

\paragraph{The performance with different backbones} We apply 3D-EMGP on SchNet~\cite{schutt2017schnet} and TorchMD-ET~\cite{tholke2021equivariant} with 7 layers, 128 hidden dimensions. The edge feature is concatenated with the RBF distance expansions. Table~\ref{tab:backbone_all} illustrates that 3D-EMGP performs consistently well on a broad family of 3D backbones.

\begin{table}[htbp]
  \centering
    \setlength{\tabcolsep}{2.5pt}
  \caption{MAE on MD17 with different backbones.}
   \resizebox{\linewidth}{!}{
    \begin{tabular}{lccccccccc}
    \toprule
    Energy      & Aspirin & Benzene & Ethanol & Malon. & Naph. & Salicylic & Toluene & Uracil & Avg. \\
    \midrule
    SchNet & 0.1502  & 0.0725  & 0.0505  & 0.0780  & \textbf{0.1103} & 0.1220  & 0.0963  & 0.1078  & 0.0985  \\
    w/ Pretrain & \textbf{0.1384} & \textbf{0.0680} & \textbf{0.0489} & \textbf{0.0721} & 0.1104  & \textbf{0.1093} & \textbf{0.0908} & \textbf{0.0996} & \textbf{0.0922} \\
    \midrule
    TorhMD-ET & 0.1050  & 0.0711  & 0.0495  & 0.0710  & 0.1120  & 0.1026  & \textbf{0.0904} & 0.1010  & 0.0878  \\
    w/ Pretrain & \textbf{0.1038} & \textbf{0.0682} & \textbf{0.0485} & \textbf{0.0707} & \textbf{0.1105} & \textbf{0.1023} & 0.0906  & \textbf{0.1003} & \textbf{0.0869} \\
    \bottomrule
    \toprule
    Force      & Aspirin & Benzene & Ethanol & Malon. & Naph. & Salicylic & Toluene & Uracil & Avg. \\
    \midrule
    SchNet & 0.4131  & 0.2374  & 0.0982  & 0.1781  & 0.1835  & 0.3145  & 0.1810  & 0.2438  & 0.2312  \\
    w/ Pretrain & \textbf{0.3186} & \textbf{0.1654} & \textbf{0.0607} & \textbf{0.1038} & \textbf{0.1229} & \textbf{0.1973} & \textbf{0.0860} & \textbf{0.1211} & \textbf{0.1470} \\
    \midrule
    TorhMD-ET & 0.1216  & 0.1479  & 0.0492  & 0.0695  & 0.0390  & 0.0655  & 0.0393  & 0.0484  & 0.0726  \\
    w/ Pretrain & \textbf{0.1124} & \textbf{0.1417} & \textbf{0.0445} & \textbf{0.0618} & \textbf{0.0352} & \textbf{0.0586} & \textbf{0.0385} & \textbf{0.0477} & \textbf{0.0676} \\
    \bottomrule
    \end{tabular}%
  }
  \label{tab:backbone_all}%
  \vskip -0.1in
\end{table}%

\paragraph{Influence of sampling steps}  We conduct EFP pretraining task with sampling steps $T=1,5,10$ in Algorithm~\ref{alg:rg_sample}, and compare the finetuning performance on MD17. As shown in Table~\ref{tab:step}, different sampling steps achieve similar results, which implies that our method is not sensitive to $T$. Hence, we select $T=1$ for simplicity and efficiency in the remaining experiments.

\begin{table}[htbp]
  \centering
    \setlength{\tabcolsep}{2.3pt}
  \caption{MAE on MD17 with different sampling steps.}
  \resizebox{\linewidth}{!}{
    \begin{tabular}{lccccccccc}
    \toprule
    Energy & Aspirin & Benzene & Ethanol & Malon. & Naph. & Salicylic & Toluene & Uracil & Average \\
    \midrule
    Base  & 0.2044  & 0.0755  & 0.0532  & 0.0748  & 0.1961  & 0.1289  & 0.1036  & 0.1165  & 0.1191  \\
    \midrule
     $T=1$ & 0.1180  & 0.0695  & 0.0485  & 0.0713  & 0.1156  & 0.1111  & 0.0903  & 0.0999  & 0.0905  \\
   $T=5$ & 0.1119  & 0.0731  & 0.0493  & 0.0731  & 0.1111  & 0.1052  & 0.0926  & 0.1025  & 0.0899  \\
     $T=10$ & 0.1194  & 0.0733  & 0.0499  & 0.0711  & 0.1120  & 0.1062  & 0.0919  & 0.1005  & 0.0905  \\
    \bottomrule
    \toprule
    Force & Aspirin & Benzene & Ethanol & Malon. & Naph. & Salicylic & Toluene & Uracil & Average \\
    \midrule
    Base  & 0.3885  & 0.1861  & 0.0599  & 0.1464  & 0.3310  & 0.2683  & 0.1563  & 0.1323  & 0.2086  \\
    \midrule
     $T=1$ & 0.1808  & 0.2211  & 0.0389  & 0.0773  & 0.0929  & 0.1773  & 0.0615  & 0.1043  & 0.1193  \\
    $T=5$ & 0.1641  & 0.2208  & 0.0401  & 0.0768  & 0.0839  & 0.1501  & 0.0669  & 0.0982  & 0.1126  \\
    $T=10$ & 0.1670  & 0.2471  & 0.0378  & 0.0758  & 0.0829  & 0.1485  & 0.0599  & 0.1009  & 0.1150  \\
    \bottomrule
    \end{tabular}%
    }
  \label{tab:step}%
  \vskip -0.1in
\end{table}%

\paragraph{Balance between node- and graph-level tasks} We attempt different loss weights for pretraining and evaluate the model performance on MD17. As shown in Table~\ref{tab:lambda}, we find that: \textbf{1. } Combination of the node- and graph-level tasks improves the performance of each individual task. \textbf{2. } Different non-zero loss weights yield comparable results, which implies that the performance of our method is not very sensitive to specific loss weights. \textbf{3. } The $\lambda_1=1.0, \lambda_2=0.2$ setting achieves relatively better performance, and we keep this setting in the remaining experiments.

\begin{table}[htbp]
  \centering
    \setlength{\tabcolsep}{2.5pt}
  \caption{MAE on MD17 energy (top) and forces (bottom) tasks with different loss weights. $\lambda_1, \lambda_2$ stand for the loss weights for EFP and INP tasks, respectively.}
  \resizebox{\linewidth}{!}{
    \begin{tabular}{ccccccccccc}
    \toprule
    $\lambda_1$ & $\lambda_2$ & Aspirin & Benzene & Ethanol & Malon. & Naph. & Salicylic & Toluene & Uracil & Average \\
    \midrule
    1.0   & 0.2   & 0.1118  & 0.0671  & 0.0484  & 0.0693  & 0.1107  & 0.1058  & 0.0874  & 0.1001  & 0.0876  \\
    0.2   & 1.0   & 0.1083  & 0.0723  & 0.0496  & 0.0717  & 0.1131  & 0.1015  & 0.0883  & 0.0971  & 0.0877  \\
    0.5   & 0.5   & 0.1088  & 0.0696  & 0.0486  & 0.0710  & 0.1153  & 0.1041  & 0.0933  & 0.1002  & 0.0889  \\
    1.0   & 0.0   & 0.1180  & 0.0695  & 0.0485  & 0.0713  & 0.1156  & 0.1111  & 0.0903  & 0.0999  & 0.0905  \\
    0.0   & 1.0   & 0.1593  & 0.0684  & 0.0494  & 0.0750  & 0.1217  & 0.1199  & 0.0892  & 0.0962  & 0.0974  \\
    \bottomrule
    \toprule
    1.0   & 0.2   & 0.1560  & 0.1648  & 0.0389  & 0.0737  & 0.0829  & 0.1187  & 0.0619  & 0.0773  & 0.0968  \\
    0.2   & 1.0   & 0.1615  & 0.1700  & 0.0384  & 0.0716  & 0.0832  & 0.1211  & 0.0635  & 0.0933  & 0.1003  \\
    0.5   & 0.5   & 0.1824  & 0.1843  & 0.0413  & 0.0737  & 0.0856  & 0.1450  & 0.0627  & 0.0985  & 0.1092  \\
    1.0   & 0.0   & 0.1808  & 0.2211  & 0.0389  & 0.0773  & 0.0929  & 0.1773  & 0.0615  & 0.1043  & 0.1193  \\
    0.0   & 1.0   & 0.2709  & 0.1704  & 0.0397  & 0.0756  & 0.1373  & 0.2184  & 0.0599  & 0.1080  & 0.1350  \\
    \bottomrule
    \end{tabular}%
    }
  \label{tab:lambda}%
  \vskip -0.15in
\end{table}%

\paragraph{Influence of the quality of the pretraining dataset.} In this work, we adapt the GEOM-QM9 as the pretraining dataset. To investigate the effect of data quality on the pretrained model performance, we attempt to add noise to the conformers of the GEOM-QM9 dataset and analyze how the pretraining dataset with noisy, low-quality conformers affects the performance of the pretrained model. For implementation, given a conformation in GEOM-QM9, we perturb its atom coordinates $\mX$ with $\mX'=\mX+ 0.1\bm{\epsilon}$ where $\bm{\epsilon}\sim \mathcal{N}(0,I_{3N})$. We apply the proposed 3D-EMGP to this perturbed dataset. Results in Table~\ref{tab:noise_qm9} and Table~\ref{tab:noise_md17} show that adding noise on coordinates does have negative impact on the performance. However, compared with Base, the model still benefits from our pretraining and consistently yields favorable results in most of the metrics.

\begin{table}[htbp]
  \centering
    \setlength{\tabcolsep}{2.5pt}
  \caption{MAE on QM9 with model pretrained on the original and noisy dataset.}
    \resizebox{\linewidth}{!}{
    \begin{tabular}{lcccccccccccc}
    \toprule
   & $\alpha$      & $\Delta_{\epsilon}$        & $\epsilon_{\text{HOMO}}$       &  $\epsilon_{\text{LUMO}}$      & $\mu$         & $C_\nu$         & $G$          & $H$          & $R^2$         & $U$          & $U_0$         & ZPVE \\
    \midrule
    Base  & 0.070 & 49.9  & 28.0  & 24.3  & 0.031 & 0.031 & 10.1  & 10.9  & 0.067 & 9.7   & 9.3   & 1.51 \\
    3D-EMGP & 0.057 & 37.1  & 21.3  & 18.2  & 0.020 & 0.026 & 9.3   & 8.7   & 0.092 & 8.6   & 8.6   & 1.38 \\
    + Noisy data & 0.064 & 40.5  & 24.0  & 20.3  & 0.024 & 0.028 & 11.1  & 10.3  & 0.124 & 10.6  & 10.5  & 1.45 \\
    \bottomrule
    \end{tabular}%
    }
  \label{tab:noise_qm9}%
  \vskip -0.15in
\end{table}%

\begin{table}[htbp]
  \centering
    \setlength{\tabcolsep}{2.5pt}
  \caption{MAE on MD17 energy (top) and forces (bottom) with model pretrained on the original and noisy dataset.}
    \resizebox{\linewidth}{!}{
    \begin{tabular}{lcccccccc|c}
    \toprule
    Energy & Aspirin & Benzene & Ethanol & Malon. & Naph. & Salicylic & Toluene & Uracil & Average \\
    \midrule
    Base  & 0.2044 & 0.0755 & 0.0532 & 0.0748 & 0.1961 & 0.1289 & 0.1036 & 0.1165 & 0.1191 \\
    3D-EMGP & 0.1118 & 0.0671 & 0.0484 & 0.0693 & 0.1107 & 0.1058 & 0.0874 & 0.1001 & 0.0876 \\
    + Noisy data & 0.1019 & 0.0715 & 0.0490 & 0.0728 & 0.1136 & 0.1067 & 0.0906 & 0.099 & 0.0881 \\
    \midrule
    Force & Aspirin & Benzene & Ethanol & Malon. & Naph. & Salicylic & Toluene & Uracil & Average \\
    \midrule
    Base  & 0.3885 & 0.1861 & 0.0599 & 0.1464 & 0.3310 & 0.2683 & 0.1563 & 0.1323 & 0.2086 \\
    3D-EMGP & 0.1560 & 0.1648 & 0.0389 & 0.0737 & 0.0829 & 0.1187 & 0.0619 & 0.0773 & 0.0968 \\
    + Noisy data & 0.1568 & 0.2374 & 0.0392 & 0.0681 & 0.0979 & 0.1247 & 0.0568 & 0.0896 & 0.1088 \\
    \bottomrule
    \end{tabular}%
    }
  \label{tab:noise_md17}%
  \vskip -0.15in
\end{table}%

\paragraph{Impact on the edge construction method} As mentioned in~\ref{sec:detail_base}, we distinguish the node connections with 1,2,3-hop and others and determine the 1-hop connection by covalent bonds. We further investigate the necessity of utilizing the bond connections. Different from our original setting, we attempt to model the molecule as a fully-connected graph and compare the performance on QM9. Results in Table~\ref{tab:bond} indicate that even without the bond connections, our pretraining method still improves the performance consistently.

\begin{table}[htbp]
\vskip -0.1in
  \centering
    \setlength{\tabcolsep}{2.5pt}
  \caption{MAE on QM9 w/ and w/o bond connections.}
  \vskip -0.1in
     \resizebox{\linewidth}{!}{
    \begin{tabular}{lcccccccccccc}
\toprule
   & $\alpha$      & $\Delta_{\epsilon}$        & $\epsilon_{\text{HOMO}}$       &  $\epsilon_{\text{LUMO}}$      & $\mu$         & $C_\nu$         & $G$          & $H$          & $R^2$         & $U$          & $U_0$         & ZPVE \\
\midrule
    Base  & 0.070 & 49.9  & 28.0  & 24.3  & 0.031 & 0.031 & 10.1  & 10.9  & 0.067 & 9.7   & 9.3   & 1.51 \\
    w/o bond & 0.070 & 48.6  & 29.3  & 25.0  & 0.030 & 0.031 & 10.8  & 10.5  & 0.075 & 10.0  & 10.7  & 1.60 \\
    3D-EMGP & 0.057 & 37.1  & 21.3  & 18.2  & 0.020 & 0.026 & 9.3   & 8.7   & 0.092 & 8.6   & 8.6   & 1.38 \\
    w/o bond & 0.059 & 38.7  & 22.4  & 19.4  & 0.022 & 0.026 & 9.2   & 8.7   & 0.085 & 8.2   & 8.2   & 1.34 \\
    \bottomrule
    \end{tabular}%
    }
  \label{tab:bond}%
  \vskip -0.15in
\end{table}%

\begin{figure}[htbp]
    \centering
    \includegraphics[width=0.45\textwidth]{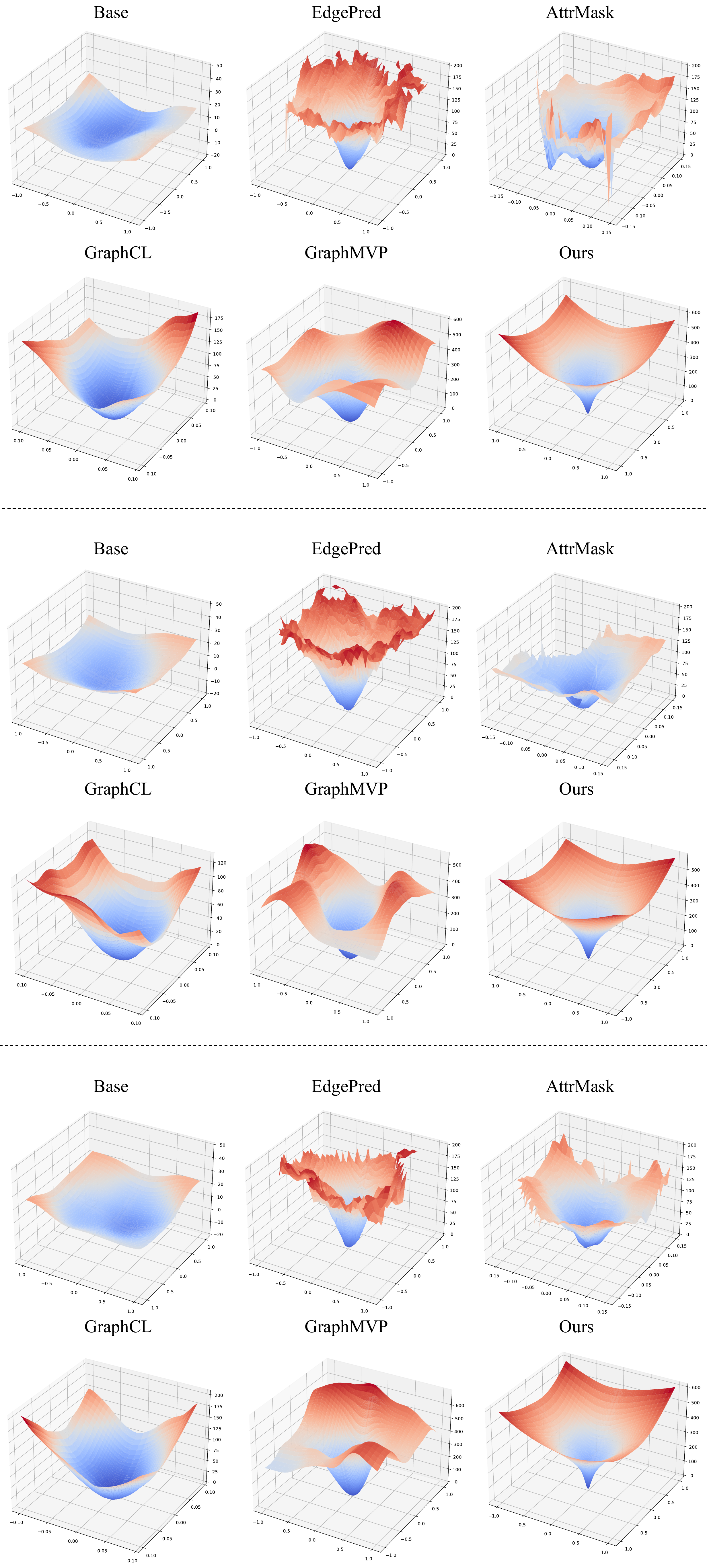}
    \caption{Additional energy landscape visualizations on MD17.}
    \label{fig:vis_more}
\end{figure}

\begin{figure}[htbp]
    \centering
    \includegraphics[width=0.45\textwidth]{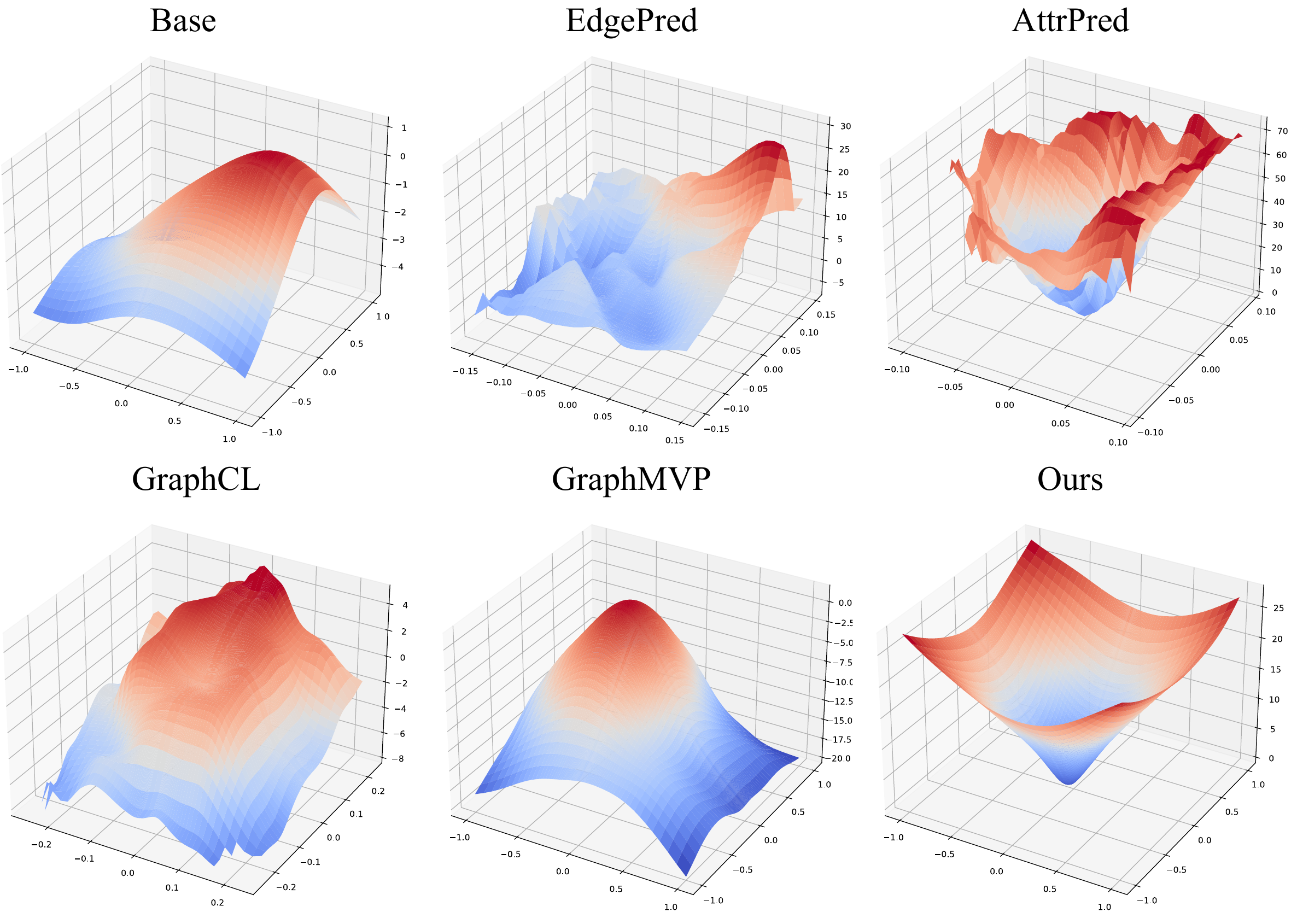}
    \caption{Energy landscape visualization on QM9.}
    \label{fig:vis_qm9}
\end{figure}

\section{More Visualizations}
\label{sec:vis}

\subsection{More Visualization on MD17}

Similar to Figure~\ref{fig:vis} in Sec.~\ref{sec:vis}, we sample 3 other direction pairs $\mD_1,\mD_2\in\R^{3\times N}$ according to Gaussian distribution and visualize the energy landscapes. Results in Figure~\ref{fig:vis_more} indicate that our pretrained model consistently provides a smooth and decreasing energy landscape.

\subsection{Visualization on QM9}

We further apply the visualization method on QM9 by finetuning an energy head on $U$ (Internal Energy at 298K) prediction task. Different from MD17, which provides a dynamic trajectory for one given molecule, QM9 provides only one metastable conformer for each molecule. As shown in Fig.~\ref{fig:vis_qm9}, the visualization results indicate that: \textbf{1. } 2D generative methods, like AttrMask or EdgePred, create irregular energy surface \emph{w.r.t.} the changes of 3D coordinates. \textbf{2. } GraphCL and GraphMVP generate upper convex surface, which contradicts with the fact that  metastable conformers stay in the local minima in the energy landscape, probably because these methods have never been exposed to the perturbed unstable conformers during the backbone pretraining and energy head finetuning stages. \textbf{3. } Our method still delivers a smooth conical surface converging towards the metastable point, which is consist with the phenomenon on MD17.

\section{Code and Related Assets}
\label{sec:code}

Our code is available at \url{https://github.com/jiaor17/3D-EMGP}. The implementation and hyper-parameter settings are partially based on ConfGF~\cite{shi2021learning}. The implementation of baseline methods refers to GraphMVP~\cite{DBLP:journals/corr/abs-2110-07728}. The implementation of backbone equivariant geometric networks refers to the official codes of EGNN~\cite{satorras2021n}, SchNetPack~\cite{doi:10.1021/acs.jctc.8b00908}, TorchMD-Net~\cite{tholke2021equivariant} and Pytorch Geometric~\cite{Fey/Lenssen/2019}. 

\end{document}